\documentclass[10pt, conference, letterpaper]{IEEEtran}
\IEEEoverridecommandlockouts 

\usepackage{mdframed}
\usepackage{diagbox}  
\usepackage{latexsym,bm}
\usepackage{xcolor}

 \usepackage{url} 
\usepackage{bbm}
\usepackage{threeparttable}
\usepackage{makecell}
\usepackage[pdftex]{graphicx}
\usepackage{graphicx}
\usepackage{epstopdf}
\usepackage{stfloats}
\usepackage{epsfig}
\usepackage{booktabs}
\usepackage{multirow}
\usepackage[ruled,linesnumbered,noend]{algorithm2e}
\usepackage{amsmath}
\usepackage{amssymb}
 \usepackage{amsfonts}
\usepackage{epstopdf}
\usepackage{subfigure}
\usepackage{graphicx}
\usepackage{array}
 \usepackage{bigstrut}
\usepackage{cite}

\usepackage{comment}
\usepackage{latexsym,bm}
\usepackage{xcolor}

\usepackage{makecell}
\usepackage[pdftex]{graphicx}
\usepackage{graphicx}
\usepackage{stfloats}
\usepackage{epsfig}
\usepackage{booktabs}
\usepackage{multirow}
\usepackage[ruled,linesnumbered,noend]{algorithm2e}
\usepackage{amsmath}
\usepackage{amssymb}
\usepackage{amsthm}
\usepackage{subfigure}
\usepackage{graphicx}
\usepackage{array}
\usepackage{comment} 
\usepackage{cite}

\newcommand{\Identity}{{\rm I\kern-.2em l}}
\newcommand{\Expect}{\mathbb{E}}

\usepackage{thmtools}

\declaretheoremstyle[%
  spaceabove=3pt,%
  spacebelow=3pt,%
  headfont=\normalfont\bfseries,%
  bodyfont=\normalfont\itshape,%
  postheadspace=0.5em,%
]{theoremstyle} 

\declaretheorem[name={Definition},style=theoremstyle]{definition}
\declaretheorem[name={Assumption},style=theoremstyle]{assumption}
\declaretheorem[name={Theorem},style=theoremstyle]{theorem}

\declaretheorem[name={Lemma},style=theoremstyle]{lemma}

\declaretheoremstyle[%
  spaceabove=0pt,%
  spacebelow=3pt,%
  headfont=\normalfont\itshape,%
  postheadspace=1em,%
  qed=\qedsymbol%
]{proofstyle}

\ifCLASSINFOpdf
 
\else

\fi

\begin{document}
\title{{
Incentive Mechanism Design for Unbiased \\Federated Learning with Randomized \\
	Client Participation}
	
	\vspace{-6mm}
}%

\author{\IEEEauthorblockN{Bing Luo\IEEEauthorrefmark{1}, Yutong Feng\IEEEauthorrefmark{2}, Shiqiang Wang\IEEEauthorrefmark{4}, Jianwei Huang\IEEEauthorrefmark{2}\IEEEauthorrefmark{3}, Leandros Tassiulas\IEEEauthorrefmark{5}}
\IEEEauthorblockA{\IEEEauthorrefmark{1}Electrical and Computer Engineering, Division of Natural and Applied Sciences, Duke Kunshan University, Kunshan, China\\
\IEEEauthorrefmark{2}School of Science and Engineering, The Chinese University of Hong Kong, Shenzhen, China\\ \IEEEauthorrefmark{3}Shenzhen Institute of Artificial Intelligence and Robotics for Society, Shenzhen, China\\
\IEEEauthorrefmark{4}IBM T. J. Watson Research Center, Yorktown Heights, NY, USA\\
\IEEEauthorrefmark{5}Department of Electrical Engineering and Institute for Network Science, Yale University, USA\\
Email: bing.luo@dukekunshan.edu.cn, yutongfeng@link.cuhk.edu.cn, shiqiang.wang@ieee.org, \\jianweihuang@cuhk.edu.cn, leandros.tassiulas@yale.edu}
\thanks{{The work of Bing Luo was supported by the AIRS-CUHKSZ-Yale Joint Postdoctoral Fellowship. 
The work of Jianwei Huang was supported by  National Natural Science Foundation of China (Project 62271434), Shenzhen Science and Technology Program (Project JCYJ20210324120011032), Guangdong Basic and Applied Basic Research Foundation (Project 2021B1515120008), Shenzhen Key Lab of Crowd Intelligence EmpoweredLow-Carbon Energy Network (No. ZDSYS20220606100601002), and the Shenzhen Institute of Artificial Intelligence and Robotics for Society. The work of Leandros Tassiulas was  supported by ARO W911NF-23-1-0088,  NSF-AoF FAIN 2132573  and  NSF  CNS-2112562. (Corresponding author: Jianwei Huang.)}}
\vspace{-0.3in}
}

\maketitle

\begin{abstract}
Incentive mechanism is crucial for federated learning (FL) when rational clients do not have the same interests in the global model as the server. However, due to system heterogeneity and limited budget, it is generally impractical for the server to incentivize all clients to participate in all training rounds (known as full participation). 
{The existing 
FL incentive mechanisms are typically designed by 
 stimulating  \emph{a fixed subset of clients} based on their data quantity or system resources. Hence, FL is performed only using this subset of clients throughout the entire training process, leading to a \emph{biased model because of data heterogeneity}.}  
This paper proposes a game-theoretic incentive mechanism for FL with \emph{randomized client participation}, where the server adopts a customized pricing strategy that motivates different clients to join with different \emph{participation levels (probabilities)} for obtaining an \emph{unbiased} and high-performance model. Each client responds to the server's monetary incentive by choosing its best participation level, to maximize its profit based on not only the incurred local cost but also its \emph{intrinsic value} for the global model. To  effectively evaluate clients' contribution to the model performance, we derive a new  {convergence bound} which analytically predicts how clients' arbitrary participation levels and their heterogeneous data affect the model performance. 
By solving a non-convex optimization problem, our analysis reveals that the intrinsic value leads to the interesting possibility of \emph{bi-directional payment} between the server and clients. 
Experimental results using real datasets on a hardware prototype demonstrate the superiority of our  mechanism in achieving higher model performance for the server as well as higher profits for the clients. 
	\end{abstract}

	\maketitle
	\section{Introduction} \label{sec:intro}

\label{sec:mot}
{Federated learning (FL)} has recently emerged as an attractive distributed machine learning paradigm, which 
enables many clients 
 to collaboratively train a machine learning model under the coordination of a central server, while keeping the training  data private \cite{kairouz2019advances}. 
In FL, each client exploits its local dataset  to compute a local model update, and the server periodically aggregates these local model updates to obtain a global model \cite{mcmahan2017communication}. 
Because clients and the server in FL usually belong to different entities, clients may not have as much interest in obtaining a high-performance  model (i.e., higher  accuracy and lower loss) as the server. For instance, a company could adopt the FL framework to train a commercial model by letting its customers   train the model with their local data and on their own devices,  while the customers (clients) may not be interested in the model. 
Hence, without {sufficient compensation}, clients may not be willing to participate due to the associated local cost (e.g., resource consumption for computation and communication), 
which makes the incentive mechanism design crucial in FL systems \cite{zhan2021survey,zeng2021comprehensive,zhang2021jsac}. 

However, 
designing an efficient and effective incentive mechanism for FL is challenging due to {two unique FL features \cite{li2018federated}:}
  1)  clients in FL systems are usually massively distributed with 
different local resources and independent availability (known as \emph{system heterogeneity});
    2) 
    the training data
are unbalanced and non-i.i.d. 
across the clients (known as \emph{statistical/data heterogeneity}). 

Due to system heterogeneity,  
it is generally impractical to design an incentive mechanism for FL that requires all clients to participate in \emph{all  
training rounds} (known as full client participation). This is because FL usually involves a large number of clients and multiple training rounds, and 
incentivizing 
full client 
participation requires 
 a considerable monetary compensation, 
 which may be beyond the server's budget. 
Moreover,  clients may be only intermittently available due to their usage patterns, which 
prevents them from participating 
in every training round. 

In the existing studies (e.g., \cite{lim2020hierarchical,kang2019incentive,ding2020optimal,ng2020joint,le2021incentive,jiao2020toward,zeng2020fmore,deng2021fair}), incentive mechanism  
with partial client participation  
mainly stimulate a \emph{deterministic} (fixed) subset of  ``valuable'' clients based on their data quantity,  computation or communication resources, where FL is performed only using this subset of clients throughout the entire training process. 
 {However, these mechanisms may 
result in a 
severely \emph{biased model} because of statistical heterogeneity (e.g., the data at incentivized clients may not be representative of all clients' data), which \emph{fails to converge to the optimal model that would be obtained if all the clients participate in training}}. 
Therefore, we are motivated to study the first key challenging question:

{\setlength{\parskip}{0.2em}\noindent\textbf{Question 1:}} \emph{How to design a practical incentive mechanism for FL with partial client participation, to ensure the convergence to a globally optimal unbiased model?}

{\setlength{\parskip}{0.2em}
In tackling the bias issue due to partial client participation, recent works have proposed 
various client sampling and model aggregation algorithms (e.g., \cite{yang2021achieving,li2019convergence, qu2020federated,chen2020optimal,rizk2020federated,9579038,cho2020client,luo2022tackling}),  
where the  
clients are presumed to be always active (available) upon the server’s request.  
This assumption may not always hold in practice, as  clients are independent decision-makers  with different local interests.  
This motivates us to 
explore {the role of 
\emph{clients' independent
participation levels (probabilities)} in incentive mechanism design, where the server offers monetary rewards to stimulate all clients to join with desired participation levels.}

From the server's perspective, intuitively, it can achieve a higher model performance  if all clients participate at higher levels, whereas the resulting payment may be beyond the limited budget. 
In addition, due to statistical heterogeneity, 
clients with higher participation levels may have low data quality (e.g., small datasize or skewed data distribution) and thus have little contribution to the model performance. 
Therefore, it is critical to design an efficient payment strategy
where clients who contribute more towards the global training objective receive more rewards. 
However, without actually training the model, it is generally impossible to determine how clients' participation level and their non-i.i.d. data affect the model performance. 
This leads to our second key question: }


{\setlength{\parskip}{0.2em}\noindent\textbf{Question 2:}} \emph{How to incentivize the clients' active participation in FL by measuring   the  contribution of each client's participation level and local (non i.i.d.) data on the model performance, so as to design an efficient payment strategy?}

{\setlength{\parskip}{0.1em}{From the clients' perspective, 
instead of purely providing training data and services to gain profit,  
clients may have  \emph{intrinsic value  (motivation)} to  participate in 
FL, e.g., to obtain 
a more powerful global model for their own use when  
they have few or skewed local data samples. 
Depending on a client's intrinsic value, the server may adopt a flexible pricing mechanism and achieve the proper trade-off between encouraging clients' participation and reducing the overall cost. 
\emph{To our best knowledge, the role of intrinsic value has not been studied in the FL incentive mechanism}.}} 

\begin{figure}[t]
	\centering
	\includegraphics[width=8.4cm]{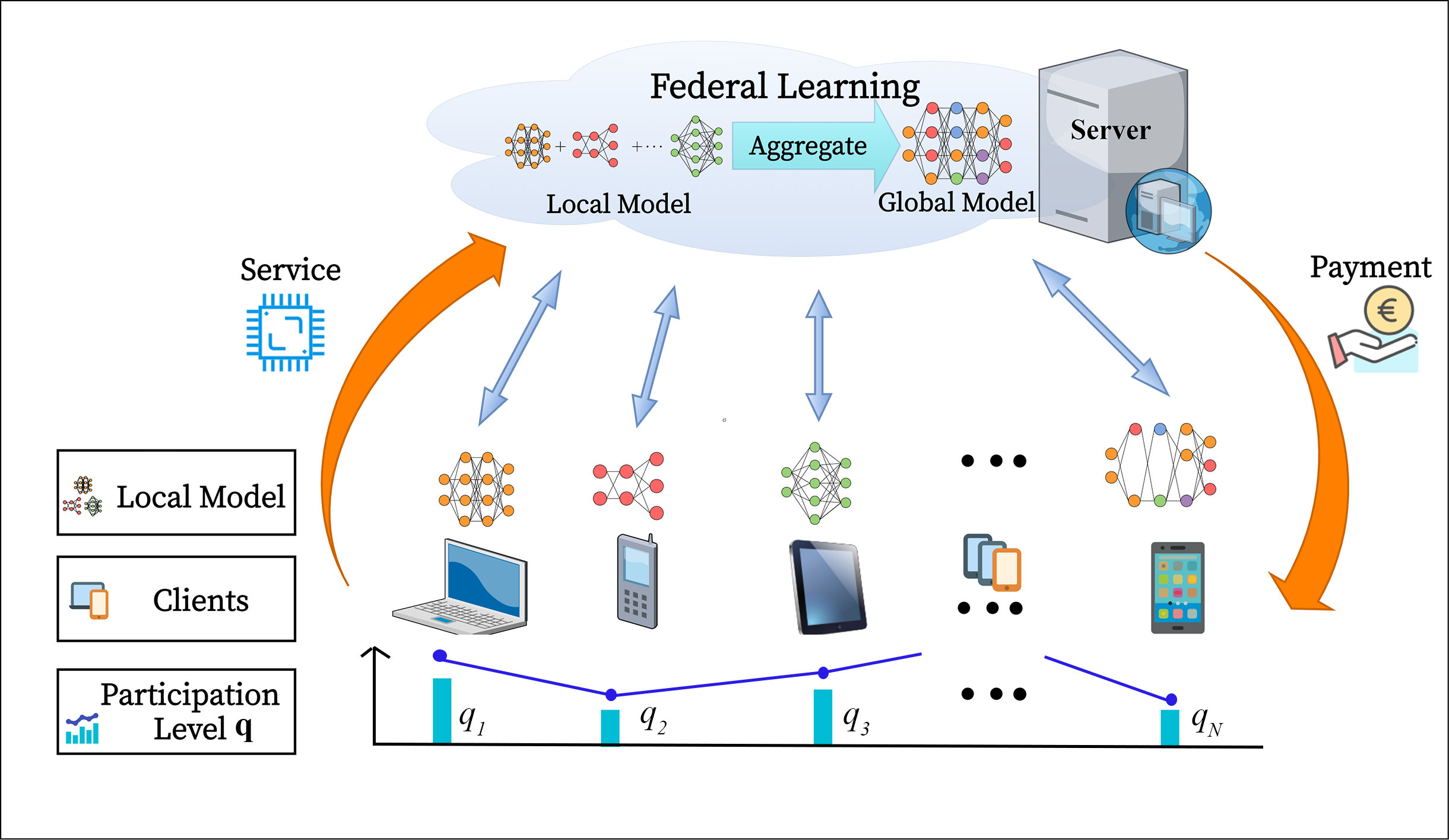}
	\vspace{-1mm}
	\caption{{Incentive mechanism for federated learning with randomized client participation, where the server adopts a customized payment strategy to motivate clients to join with different
participation levels $\boldsymbol{q}$.}}
	\label{model_fig}
	\vspace{-2mm}
\end{figure}


 {\setlength{\parskip}{0.4em}In light of the above discussion, this paper proposes a novel incentive mechanism for FL with \emph{randomized partial client participation},   where the server motivates clients to join FL with different participation levels (probabilities) as shown in Fig.~\ref{model_fig}. 
Specifically, we model the interaction 
between the server and clients as a sequential two-stage \emph{Stackelberg game}, where the server (leader) decides a \emph{customized pricing} scheme for each client's participation level to  maximize the global model performance. 
Each client (follower) then responds to the server's monetary incentive by independently 
choosing its best participating level to maximize its own profit,  based on 
the incurred local cost as well as the intrinsic value for the global model.} 

{\setlength{\parskip}{0.1em}We first develop an adaptive model aggregation scheme 
to guarantee the model unbiasedness for arbitrary independent client participation levels. With the unbiased design, we obtain a new FL  convergence bound, which analytically establishes the relationship between the global model performance and clients’ arbitrary participation levels with unbalanced and non-i.i.d. data {without actually training the model.}}

{\setlength{\parskip}{0.1em}Then, based on the convergence bound, we obtain the Stackelberg Equilibrium (SE) via solving a non-convex optimization
problem, which reveals interesting  design principles. Counter-intuitively, we show that in order to achieve an unbiased global model, the server will set higher prices for clients who have larger local costs. For the clients, 
those with higher intrinsic values may receive lower prices (compensations) from the server and choose lower participation
levels. When the clients’ intrinsic values are higher than a certain threshold, the clients may need to pay for the server for participation, which leads to a \emph{bi-directional payment}.} 

{\setlength{\parskip}{0.3em}Finally, we evaluate the result of our proposed game with both real and synthetic datasets on a hardware prototype. Experimental results demonstrate the superiority of our proposed mechanism in achieving higher global model performance for the server as well as higher profits for the clients. For example, for MNIST dataset, our proposed pricing spends $69\%$ less time than the baseline uniform pricing for reaching the same target accuracy under the same budget.}

{In summary, our key contributions are as follows: }
\begin{itemize}
\item {\emph{Unbiased Incentive Mechanism for FL with Randomized Client Participation:} 
To the best of our knowledge, 
we are the first to study 
FL incentive mechanism with practically randomized client participation that guarantees the unbiasedness of the obtained global model towards full client participation.}

\item {\setlength{\parskip}{0.3em}{\emph{Convergence Bound for Evaluating Clients’ Contribution:} 
We obtain a new FL convergence upper bound for 
clients’ arbitrary participation levels with  unbalanced and non-i.i.d. data, which allows us to effectively measure each client's contribution to the global model performance.} }


\item {\setlength{\parskip}{0.3em}{\emph{Intrinsic Value Design:} 
We prove that the intrinsic
value can lead to a flexible bi-directional payment between the server and clients. As far as we know, this is the first work that models the impact of clients’ intrinsic value in the FL incentive mechanism.}}

\item {\setlength{\parskip}{0.3em}\emph{Insightful Equilibrium Properties:} 

We obtain insightful design principles that characterize the impact of the server's and clients' system parameters. In particular, 
%
 our analysis reveals a counter-intuitive  result that the sever may set higher prices for clients who have larger local costs in order to guarantee the model unbiasedness. }

\item {\setlength{\parskip}{0.3em}\emph{Experimental Evaluation with Hardware Prototype}: We provide hardware experimental results that 
demonstrate the superiority of our proposed mechanism 
with a $69\%$ time savings compared to baseline uniform pricing for MNIST dataset under the same budget.} 

\end{itemize}

\section{Related Works}
{Incentive mechanism design for FL has received significant attention over recent years (for a comprehensive review please refer to \cite{tu2021incentive}). 
In the literature,  existing works  mainly focus on designing incentives mechanism  with full client participation (e.g.,  \cite{sarikaya2019motivating,khan2020federated,9488705, zhan2020incentive}). }

{\setlength{\parskip}{0.3em}Considering the server's limited budget and clients' heterogeneous local resources,
%
recent FL incentive mechanisms (e.g., \cite{lim2020hierarchical,kang2019incentive,ding2020optimal}) have considered more 
practical partial client participation scenarios.
However, these works 
mainly stimulate a \emph{deterministic} subset of “valuable'' clients  
and train the global model only with these clients' data. 
Although such mechanisms may speed up the training process at a certain level and have merits in eliciting clients' private information, they 
may result in a biased global model if the selected clients have skewed data. 
We propose to incentivize all clients to join FL with independent participation probabilities, which enables the server to obtain all clients' data contributions 
to guarantee the model's unbiasedness.} 


{\setlength{\parskip}{0.3em}In addition, most existing incentive mechanisms evaluate clients' contribution based on their data quantity (e.g.,  \cite{lim2020hierarchical,kang2019incentive, zeng2020fmore,ding2020optimal,jiao2020toward}) or computation/communication resources (e.g.,  \cite{jiao2020toward,sarikaya2019motivating,khan2020federated,zhan2020incentive,ding2020optimal,ng2020joint,le2021incentive,deng2021fair}), which do not capture the critical heterogeneous data distribution across clients. Our mechanism measures clients' contribution via a theoretical convergence result for general statistical heterogeneity with unbalanced and non-i.i.d. data.} 


The organization of the rest of the paper is as follows.
Section III introduces federated learning and game model.
Section IV presents our new error-convergence bound
with randomized client participation level. Section V analyzes  the equilibrium of the proposed game 
 and solution insights. Section VI
provides hardware-based cross-device  experimental results.
We conclude this paper in Section VII.

\section{Federated Learning  and Game Model}
\label{sec:systemModel}
We first introduce the basics of FL with partial client participation and the proposed unbiased and independent client participation scheme in Section~\ref{2.1}. Then, we  present the incentive mechanism design for  client participation level (probability), and describe the decision problems for the server and clients, respectively in Section~\ref{2.2}. Finally,    Section~\ref{2.4} presents the strategic interactions between the server and clients with the 
proposed Stackelberg game as well as the challenges in solving the game. 


\subsection{FL with Randomized Partial Client Participation}
\label{2.1}
We consider a typical FL scenario, where a central server wants to learn a model based on the data from a set  $\mathcal{N}=\{1,\ldots, N\}$ clients. Each client $n \in \mathcal{N}$ has $d_n$ data samples,  ($\boldsymbol{x}_{n, 1}, \ldots, \boldsymbol{x}_{n, d_{n}}$), which is distributed in a non-i.i.d. fashion. 
Define ${{f}\left( \boldsymbol{w}; \boldsymbol{x}_{n, i} \right)}$ as the loss function,  indicating how the machine learning model parameter $\boldsymbol{w}$ performs on the data sample $\boldsymbol{x}_{n, i}$. Thus, the local loss function of client $n$ is
\begin{equation}
\label{lo_ob}
{F_n}\left( \boldsymbol{w} \right) := \frac{1}{{{d_n}}}\sum\nolimits_{i =1}^{d_n} {{f}\left( \boldsymbol{w}; \boldsymbol{x}_{n, i} \right)}.
\end{equation}
We further denote $a_n={d_n}/{\sum\nolimits_{n = 1}^N \!d_n}$ as the weight of the $n$-th device, where $\sum\nolimits_{n = 1}^N a_n=1$. By denoting $F\left( \boldsymbol{w} \right)$ as the global loss function, the goal of FL is to solve the following optimization problem 
\cite{kairouz2019advances}:
\begin{equation}
\label{gl_ob}
\min_{\boldsymbol{w}}  F\left( \boldsymbol{w} \right) :=\sum\nolimits_{n = 1}^N{a_n}{F_n}\left( \boldsymbol{w} \right).
\end{equation}


Due to limited system bandwidth and clients' diverse availability, the most popular and de facto optimization algorithm to solve \eqref{gl_ob} is FedAvg \cite{mcmahan2017communication}. 
In FedAvg, the server randomly samples a fraction of $K$ clients (known as \emph{partial client participation}) in each round, and each selected client performs multiple (e.g., $E$) steps of local stochastic gradient descent (SGD) iterations on \eqref{lo_ob}. Then, the server 
aggregates their resulting local model updates periodically for a given deadline of  $R$ rounds or until the global loss \eqref{gl_ob} converges.

However, considering that clients in FL are independent decision-makers, 
each {client $n$ can  
decide its own participation level (probability)} $q_n$, instead of using the sampling probability decided by the server. 
Nevertheless, \emph{without a careful algorithm design, the obtained global model can be severely biased due to statistical heterogeneity}. 

In the following, 
before we present our incentive mechanism design,  we first propose an unbiased model aggregation scheme under an arbitrary independent client participation level $\boldsymbol{q}\!=\!\{q_1, \ldots, q_N\}$, such that the obtained model based on our mechanism is unbiased towards full client participation.
We define the weighted aggregated model with full client participation for any round $r$ as
    $\overline{\boldsymbol{w}}^{r+1}:=\sum\nolimits_{n=1}^Na_n\boldsymbol{w}_n^{r+1}$.
With this, we have the following result. 

{\setlength{\parskip}{0.3em}\begin{lemma}
\label{adaptive_sam_agg}
{(Unbiased FL with  Independent Client Participation Level)} 
For clients under an arbitrary  participation level $\boldsymbol{q}$, 
we aggregate the participants' local updates as 
\begin{equation}
   \label{aggregation}
 \boldsymbol{w}^{r+1} \leftarrow \boldsymbol{w}^{r}+\sum_{n \in \mathcal{S}(\boldsymbol{q})^{r}} \frac{a_{n}}{q_{n}} \left(\boldsymbol{w}_n^{r+1}-\boldsymbol{w}^{r}\right),
\end{equation}
where $\mathcal{S}(\boldsymbol{q})^{r}$ is  the participating client set in round $r$. Then, we have  
\begin{equation}
    \label{unbiased_agg}
    \mathbb{E}_{\mathcal{S}(\boldsymbol{q})^{r}}[\boldsymbol{w}^{r+1}]= \overline{\boldsymbol{w}}^{r+1}.
\end{equation}
\end{lemma}}
\begin{proof}
We give the proof in Appendix~A. 
\end{proof}

{\setlength{\parskip}{0.1em}\textbf{Remark:}
{The interpretation
of our {aggregation} scheme is similar to that of \emph{importance sampling}. More specifically, we inversely re-weight  the participant's updated model gradient in the aggregation step (e.g.,  $\frac{1}{q_n}$ for client $n$),  such that the aggregated  model is unbiased towards the true update with full client participation}. {{{We note that simply inversely weighting the model updates from the sampled clients does not yield an unbiased global model, i.e. $ \Expect_{\mathcal{K}(\boldsymbol{q})^{(r)}}[\sum_{i \in \mathcal{K}(\boldsymbol{q})^{(r)}} \frac{p_{i}}{K q_{i}}\mathbf{w}_i^{(r+1)}] \neq \overline{\mathbf{w}}^{(r+1)}$, as the equality holds only when clients are sampled uniformly at random, i.e., ${q_i}=1/N$.}}} Particularly, when $q_n=1$ for all $n$, $\mathcal{S}(\boldsymbol{q})^{r}$ is the full set with all $N$ clients, and $\boldsymbol{w}^{r+1}$ in \eqref{aggregation}  recovers  $\overline{\boldsymbol{w}}^{r+1}$. {Nevertheless, unlike most active sampling schemes (e.g.,  \cite{yang2021achieving,li2019convergence, qu2020federated,luo2022tackling}) where clients' sampling probabilities $q_n^{\text{sam}}$ are dependent with $\sum_{n=1}^Nq_n^{\text{sam}}=1$, 
the clients' participation levels ${q_n}$ in our model is independent from  each other with the sum $\sum_{n=1}^Nq_n$ varying between $0$ to~$N$.}} 


\subsection{Mechanism Design for Randomized Client Participation
}\label{2.2} 
 
As clients are independent decision-makers, we will 
explore the impact of 
\emph{clients' independent
participation levels (probabilities)} $\boldsymbol{q}=\{q_1, \ldots, q_N\}$ in the incentive mechanism design. 
Specifically, under a  limited payment budget, the server designs a customized
pricing scheme for each client’s participation level to maximize
the model performance. 
In the following, 
we present the decision problems for the server and the clients.

\vspace{1mm}
\subsubsection{\textbf{Server's  Decision 
Problem}} 
\label{2.21}

The goal of the server is 
to minimize the training loss defined in \eqref{gl_ob}  
for a certain number of training rounds.  
To achieve this, the server imposes a set of prices $\boldsymbol{P}=\{P_1, \ldots, P_N\}$ to  incentivize each client's independent participation level (probability) $\boldsymbol{q}=\{q_1, \ldots, q_N\}$ under a  payment budget $B$, where $P_i$ represents the price per unit of client $i$'s participating level. Hence, the payment for each client $n$ is $P_nq_n$. 
 
Let us denote $\boldsymbol{w}^{R}(\boldsymbol{q})$ as the obtained model after $R$ rounds when clients participate with level $\boldsymbol{q}$ under pricing strategy $\boldsymbol{P}$. 
We can formulate the server's  problem as the following \textbf{P1}:
\begin{subequations}\label{U_s}
\begin{align}
 \textbf{P1:} \ \ \min_{\boldsymbol{P}} \  & U_s(\boldsymbol{P, q}):= \mathbb{E}\left[F\left(\boldsymbol{w}^{R}(\boldsymbol{q})\right)\right], \label{U_s_a}\\ 
         \text{s.t.} 
           &\ \  \sum\nolimits_{n=1}^NP_nq_n \le B.\label{U_s_b}
         \end{align}
\end{subequations}
The expectation of the objective in \eqref{U_s_a} comes from the randomness in client's participation level $\boldsymbol{q}$ and local SGD.
 Since the total budget is limited, it is important for the server to design the optimal pricing strategy $\boldsymbol{P}$ to maximize its utility. 

 \emph{\textbf{Remark}:  Unlike existing incentive mechanisms in FL where the server always pays for the clients (e.g., $P_n\ge0$, for all $n$) \cite{lim2020hierarchical,kang2019incentive,ding2020optimal,ng2020joint,le2021incentive,jiao2020toward,zeng2020fmore,deng2021fair},  we allow the payment to be bi-directional. 
 This is because, instead of purely providing training data and computing services, if some client $n$ has a high appreciation for the global model as we will show later, 
 it may be willing to pay for the server (i.e., $P_n<0$).} 
 
\vspace{1mm}
\subsubsection{\textbf{Clients' Decision Problem}}\label{2.3}


Each client's goal is to choose its participation level $q_n$ to maximize its utility function $U_n$, based on its incurred \emph{local cost} and its \emph{intrinsic value} for the global model, which we will explain next.

\emph{{Local Cost Model}}.  The cost of client $n$ involves resource consumption for model computation and communication as well as   the lost opportunity for joining other activities for monetary reward.  Intuitively, the higher participation level the higher cost will be, so we model the cost function $C_n$ as
\begin{equation}
    \label{cost}
    C_n= c_nq_n^{\tau}, \ \tau>1, \ 0\le q_n\le q_{n,\text{max}}. 
\end{equation}
The exponent  $\tau>1$ 
captures a broad class of convex cost functions, indicating an increasing rate as $q_n$ increases. 
We let $\tau=2$ for analytical tractability in the rest of the paper, which is also a standard assumption in economic models when the decision variable is constrained \cite{candogan2012optimal}, e.g., $q_n\le q_{n,\text{max}}$.
Nevertheless, we claim that our theoretical results in this paper also hold for an arbitrary $\tau>1$.  
Parameter $c_n>0$ is the local cost parameter.

\emph{{Intrinsic Value Model}}. In addition to incurring the resource cost,  clients may have intrinsic motivation to participate in FL, e.g., obtaining the powerful global model. In order to effectively model the intrinsic value, we 
denote $F({\boldsymbol{w}_n^*})$ as the loss when client $n$ applies its locally optimal model $\boldsymbol{w}_n^*$ 
into the global loss function \eqref{gl_ob}, where $\boldsymbol{w}_n^*=\arg \min_{\boldsymbol{w}}{F_{n}(\boldsymbol{w})}$ is solved using client's local training data on its local loss function ${F_{n}(\boldsymbol{w})}$  as defined in \eqref{lo_ob}. Then, we model the intrinsic value $V_n$ for client $n$ as 
\begin{equation}
    \label{intrinsic}
   V_n:= v_n\left(F({\boldsymbol{w}_n^*})-\mathbb{E}\left[F\left(\boldsymbol{w}^{R}(\boldsymbol{q})\right)\right]\right).
\end{equation} 
The value of $F({w_n^*})-\mathbb{E}\left[F\left(\boldsymbol{w}^{R}(\boldsymbol{q})\right)\right]$ represents the improvement of model performance due to the participation in FL.  Parameter $v_n\ge0$ is the preference level for the improvement, since clients may have different preferences for the same model improvement.  
Intuitively, 
given that $F({\boldsymbol{w}_n^*})$ is a constant value and independent of $\boldsymbol{q}$, a lower 
$\mathbb{E}\left[F\left(\boldsymbol{w}^{R}(\boldsymbol{q})\right)\right]$ yields a higher intrinsic value $V_n$, \emph{indicating that client $n$ has an internal drive to minimize the server's utility in \eqref{U_s_a}}.  

Based on the above, we formulate 
each client $n$'s decision problem as follows:
\begin{subequations}
\label{U_i}
\begin{align}
  \textbf{P2}:      \max_{{q_n}}    U_n({q_n, P_n})&:={P_nq_n}
-{C_n}+{V_n}
\label{U_i_a}\\
     \text{s.t.}  \   \quad 0 &\le q_n\le q_{n,\text{max}}. 
     \end{align}
\end{subequations}

\emph{\textbf{Remark:} As we discussed in the  previous subsection, due to the existence of the intrinsic value, clients may have incentives to participate in FL even without monetary reward, i.e., $P_n=0$. In certain cases, when some client $n$ has a very high intrinsic value $v_n$, 
it is possible for client $n$ to pay for the server for participation, i.e., $P_n<0$.}  

\begin{figure}[t]
	\centering
	\includegraphics[width=8.5cm]{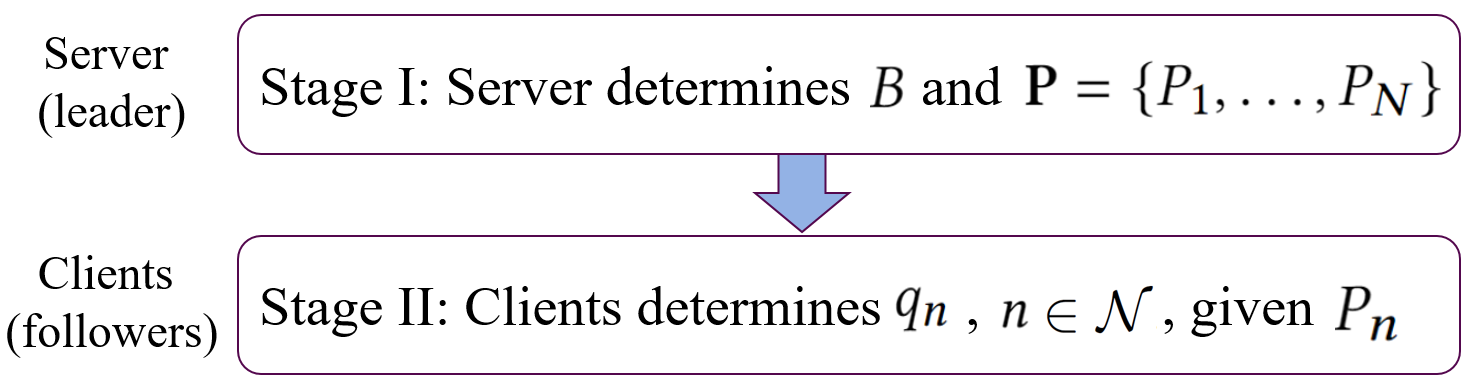}
	\caption{Stackelberg game between the server and Clients.}
	\label{stack}
	\vspace{-2mm}
\end{figure}

\subsection{Stackelberg Game Formulation}\label{2.4} 
As shown in Fig.~\ref{stack}, we model the sequential decision-making between the server and clients  as a two-stage Stackelberg game  \cite{bacsar1998dynamic},  
where the server acts as the Stackelberg leader and decides the pricing variables $\boldsymbol{P}=\{P_1, \ldots, P_N\}$ to minimize its utility defined in \eqref{U_s_a} in Stage~I. 
Then, given the server's pricing strategy $\boldsymbol{P}$, each client acts as a Stackelberg follower and chooses its reactive participation level $q_n$ 
to maximize its utility defined in \eqref{U_i_a} in Stage~II. 
In the following, \emph{we refer to the proposed  Stackelberg game as Client Participation Level Game (CPL Game)}.\footnote{{{Our mechanism assumes complete information, focusing on evaluating clients’ contribution and designing an effective payment scheme. For the more realistic incomplete information
scenario, we can adopt \emph{Bayesian} method to model and analyze the performance similarly with a higher complexity.}}}

\vspace{1mm}
\subsubsection{\textbf{Solution Concept of the Proposed CPL Game}} 

The common solution concept of the  CPL Game is 
Stackelberg equilibrium (SE), which we define as follows.
\begin{definition}
    The Stackelberg equilibrium (SE) 
   of the CPL Game is a set of decisions $\{\boldsymbol{P}^{{\text{SE}}},\boldsymbol{q}^{\text{SE}}\}$ satisfying 
\begin{subequations}
\begin{align}
 q_n^{\text{SE}}(\boldsymbol{P}) &= \arg \max_{q_{n}(\boldsymbol{P})} {U_{n}}\left(q_n(\boldsymbol{P})\right), \forall n \in \mathcal{N},  \label{Nash_1} \\ 
\boldsymbol{P}^{\text{SE}}& = \arg \min_{\boldsymbol{P}} U_{s}\left(\boldsymbol{P},\boldsymbol{q}^{\text{SE}}(\boldsymbol{P})\right).\label{Nash_2} 
\end{align}
\end{subequations}
\end{definition}
At a SE, neither the server or the clients has incentive to deviate 
for better choice. 
A powerful technique to obtain SE is backward induction \cite{Microeconomic}, where we first solve for clients' decision-making $\boldsymbol{q(P)}$ given the server’s pricing scheme $\boldsymbol{P}$ in Stage II, and then move back to Stage I to determine the server's pricing strategy $\boldsymbol{P}$.


\vspace{1mm}
\subsubsection{\textbf{Challenges in Solving  the CPL Game}}\label{challenge} 
Solving the CPL Game, however, is challenging due to \emph{the lack of an 
analytical expression of $\mathbb{E}[F\left(\boldsymbol{w}^R(\boldsymbol{q})\right)]$ to characterize the impact of $\boldsymbol{q}$.} 
Hence, 
it is difficult for the server to evaluate the clients' contributions and make an efficient pricing decision.
Moreover, in general, without actually training the model, it is impossible to find out how $\boldsymbol{q}$ affects the final model $\boldsymbol{w}^R(\boldsymbol{q})$ and the corresponding  loss $\mathbb{E}[F\left(\boldsymbol{w}^R(\boldsymbol{q})\right)]$.

We will show how to address the above challenge in Section~\ref{sec:convergence}, 
and then we will find the SE solution for the CPL Game in Section~\ref{CPL_analysis}.

\section{
Convergence Analysis for Randomized Client Participation Level}
\label{sec:convergence}
In this section, we address the challenge mentioned in Section \ref{challenge} by deriving a new tractable error-convergence bound.  The convergence bound  establishes the analytical relationship between $\boldsymbol{q}$ and $\mathbb{E}[F\left(\boldsymbol{w}^R(\boldsymbol{q})\right)]$, which allows us to 
approximate 
$\mathbb{E}[F\left(\boldsymbol{w}^R(\boldsymbol{q})\right)]$ in  $U_s$ and $U_n$ and analytically solve the PCL Game.

\subsection{Key Assumptions}
\begin{assumption}
For each client $n \in \mathcal{N}$, $F_{n}$ is  $L$-smooth and $\mu$-strongly convex.
\end{assumption}
\begin{assumption}
For each client $n \in \mathcal{N}$, the stochastic gradient of $F_{n}$ is unbiased with its variance bounded by $\sigma_n^{2}$. 
\end{assumption}
\begin{assumption}
For each client $n \in \mathcal{N}$, the expected squared norm of its stochastic gradient is bounded by $G_n^{2}$.
\end{assumption}
Assumptions 1 and 2 are 
commonly made in many existing studies of  convex FL problems, such as  $\ell_{2}$-norm regularized linear regression, logistic regression (e.g., \cite{li2019convergence,yu2018parallel,chen2020optimal,stich2018local,qu2020federated,cho2020client}).  
Nevertheless, in Assumption~3, we make an 
assumption for each client n with the bound $G_n$ instead of the bound $G$ for all the clients as in 
\cite{li2019convergence,yu2018parallel,chen2020optimal,stich2018local,qu2020federated,cho2020client}. {This is because if clients’ data are i.i,d., then $G_n$ would be the same across the clients since each client locally performs SGD from the same data distribution. However, when clients have non-i.i.d. data distribution, the values of $G_n$ would be different, which not only characterizes the data heterogeneity in our convergence result but also yields a more accurate pricing design (as we will show in Section~\ref{CPL_analysis}). In practice, we can estimate $G_n$ by letting the participated clients send back their actual local stochastic gradient norms computed along the trajectory of the model updates. }

\subsection{{Bounded Model Variance} 
}
We first present 
the introduced variance of the aggregated model $\boldsymbol{w}^{r+1}$ in \eqref{aggregation} due to randomized  client participation. 
\begin{lemma}\label{thevariance}
The variance between the aggregated  model $\boldsymbol{w}^{r+1}$ in \eqref{aggregation}  and the  global model with full participation $\overline{\boldsymbol{w}}^{r+1}$ is bounded as
\begin{equation}
    \label{variance}
     \begin{array}{l}           \!\!\mathbb{E}_{\mathcal{S}(\boldsymbol{q})^{r}}\!\left\|\boldsymbol{w}^{r\!+\!1}\!-\!\overline{\boldsymbol{w}}^{r\!+\!1}\right\|^{2} \leq 4\!\sum\limits_{n=1}^N\!\dfrac{\left(1\!-\!q_n\right)a_n^2G_n^{2}}{q_n}\left({\eta^{r}} E\right)^{2}.\!\!\!
        \end{array}
\end{equation}
\end{lemma}
\noindent\emph{Proof  Sketch.}
The proof follows a similar argument  of Lemma~5 in \cite{li2019convergence}, except that our client participation levels ${q_n}$ are independent among each other, which does not satisfy the assumption of $\sum_{n=1}^Nq_n\!=\!1$ made in \cite{li2019convergence}. 
\qquad \qquad   \  \qquad \qedsymbol

{\setlength{\parskip}{0.3em}\textbf{Remark}: When $q_n=1$ for all $n$, the variance in \eqref{variance} is equal to zero, since the aggregated model $\boldsymbol{w}^{r+1}$ in the left hand side of \eqref{variance} is the same as 
$\overline{\boldsymbol{w}}^{r+1}$ because $q_n = 1$ implies full participation.}   

\subsection{{Main Convergence Result}}

\begin{theorem}
\label{convergencebound}
{(Convergence Upper Bound under an Arbitrary $\boldsymbol{q}$)} 
Consider any given client participation level $\boldsymbol{q}=\{q_1, \ldots, q_N\}$ and the unbiased aggregation in Lemma~\ref{adaptive_sam_agg}, if we choose   
the decaying learning rate $\eta_{r}=\frac{2}{\max \{{8 L}, \mu E\}+\mu r}$,  
the optimality gap after $R$ rounds satisfies 
\begin{equation}
    \label{convergence}
    \begin{array}{c}
    \mathbb{E}[F\!\left(\boldsymbol{w}^R(\boldsymbol{q})\right)]\!-\!F^{*} \!\le\dfrac{1}{R}\!
     \left(\alpha{\sum\limits_{n=1}^N\dfrac{\left(1\!-\!q_n\right)a_n^2G_n^{2}}{q_n}}\! +\!\beta\right)\!,\!\!
    \end{array}
\end{equation}
where $\alpha=\frac{8LE}{\mu^2}$,  $\beta=\frac{2L}{\mu^2E}A_0+\frac{12L^2}{\mu^2E}\Gamma+ \frac{4L^2}{\mu E}\left\|\boldsymbol{w}_{0}-\boldsymbol{w}^{*}\right\|^{2}$,  $A_0\!=\!\sum\limits_{n=1}^{N}\! a_{n}^{2} \sigma_{n}^{2}\!+\!8\sum\limits_{n=1}^Na_nG_n^2(E\!-\!1)^{2}$, and $\Gamma\!=\!F^{*}\!-\!\sum_{n=1}^{N} a_{n} F_{n}^{*}$.  
\end{theorem}
\begin{proof}
We give the proof in Appendix B.
\end{proof}

We summarize the key insights of Theorem~\ref{convergencebound} as follows:
\begin{itemize}
    \item The  bound in \eqref{convergence} is \emph{valid for  arbitrary independent client participating level $\boldsymbol{q}$ 
  and variant number of participating clients in each round} (i.e., $\sum_{n=1}^Nq_n$ can vary between $0$ to $N$), 
which is the key difference from existing  convergence results with active client sampling \cite{luo2022tackling,yang2021achieving, li2019convergence,qu2020federated,rizk2020federated,chen2020optimal,cho2020client}. 

\item The bound in \eqref{convergence} characterizes how randomized partial client participation (i.e., $q_n<1$) worsens the convergence rate compared to full client participation. 
It also indicates that in order to obtain an unbiased global model, all clients need to participate with non-zero probability for model convergence, i.e.,  $q_{n}>0$, for all $n$. This is because when $q_{n}\rightarrow0$, it will take infinite number of rounds $R$ for convergence. Our bound also explains why only incentivizing and training part of the clients in existing mechanisms (e.g., \cite{lim2020hierarchical,kang2019incentive,ding2020optimal,ng2020joint,le2021incentive,jiao2020toward,zeng2020fmore,deng2021fair}) may fail to converge to the optimal global model.  
\item The convergence bound in \eqref{convergence} \emph{establishes the relationship} between the expected loss function $\mathbb{E}[F\left(\boldsymbol{w}^R(\boldsymbol{q})\right)]$, clients’ participation level $\boldsymbol{q}$, and their heterogeneous data  $a_nG_n$. In other words, how clients' unbalanced data ($a_n$) and non-i.i.d. data distribution ($G_n$) affect the model training. This not only provides analytical utility expressions for the server $U_s$ and clients $U_n$, but also enables an efficient pricing strategy as we will show in Section~\ref{CPL_analysis}.


\end{itemize}

{The use of  $G_i$ in our convergence analysis (instead of more accurate instantaneous gradient norms) is mainly due to the challenging ``chicken and egg'' problem. Specifically,  before the model has been fully trained, it is generally impossible to know exactly how different FL configurations (e.g., clients'  different participation levels in this paper) affect the FL performance.  Therefore, to optimize the FL performance, we need to have a lightweight surrogate that can (approximately) predict what will happen if we choose a specific configuration. A common surrogate used for this purpose is the convergence upper bound, which has become a common practice in the communications and networking community (e.g., \cite{wang2019adaptive, luo2022tackling, cui2022optimal, shi2020device, perazzone2022communication, luo2020cost}). Moreover, our experiments in Section~VI demonstrate the superiority of our designed
pricing scheme based on the convergence bound, in terms
of achieving higher global model performance and higher client
profits compared to baseline schemes.}

 \section{Stackelberg Equilibrium Analysis}
\label{CPL_analysis}
In this section, 
we use the obtained convergence bound in Theorem~\ref{convergencebound} to approximate $\mathbb{E}[F\left(\boldsymbol{w}^R(\boldsymbol{q})\right)]$ in the server's utility $U_s$  and clients' utility $U_n$, and solve the CPL Game via backward induction.  We first solve clients' decision-making $\boldsymbol{q}$ given the server’s pricing scheme $\boldsymbol{P}$ in Stage II, and then move back to Stage I to determine 
$\boldsymbol{P}$. Finally, we characterize the key  insights of the optimal solution. 

\subsection{Client's Decision at Stage II}
We approximate  $\mathbb{E}[F\left(\boldsymbol{w}^R(\boldsymbol{q})\right)]$ with the convergence bound in \eqref{convergence}, and rewrite  client $n$'s problem given the server's pricing strategy $\boldsymbol{P}$ as the following Problem \textbf{P2}$^{\prime}$:
\begin{subequations}\label{U_i2}
\begin{align}
\!\!\!\!\!\!   & \textbf{P2}^{\prime}:   \  \max_{{q_n}}   \ \ U_n({\boldsymbol{q}, P_n})={P_nq_n}-{c_nq_n^2} \notag
        \\ &+{v_n\left[F({w_n^*})\!-\!F^{*}\! -\!\frac{1}{R}
     \!\left(\alpha{\sum\limits_{n=1}^N\frac{\left(1\!-\!q_n\right)a_n^2G_n^{2}}{q_n}} \!+\!\beta\right)\right]},\label{U_i2_a}\\ 
   & \qquad \text{s.t.}     \quad 0 \le q_n\le q_{n,\text{max}}.  \label{U_i2_b}
     \end{align}
  \end{subequations}
We observe that the objective function  \eqref{U_i2_a} is concave in $q_n$.  
Along with the linear constraints in \eqref{U_i2_b}, we conclude that Problem \textbf{P2}$^{\prime}$ is \emph{concave}. 
Therefore, the optimal solution of Problem  \textbf{P2}$^{\prime}$ is unique,  
which is the best choice of $q_n$ to maximize its own utility.

Based on the first order condition, 
the optimal choice of $q_n^*(P_n)$ for client $n$ satisfies
 \begin{equation}\label{U_i_analysis}
\begin{aligned}
         P_n+v_n\frac{\alpha}{R}
           \frac{a_n^2G_n^2}{q_n^{*2}}-2c_nq_n^*=0. 
     \end{aligned}
  \end{equation}
Although the closed form solution of $q_n^*(P_n)$ is complicated because \eqref{U_i_analysis} is a cubic equation, we can show that 
\emph{$q_n^*(P_n)$ is a monotonically increasing convex function in $P_n$}. 

Based on the client's optimal solution in \eqref{U_i_analysis}, we move to Stage I of the CPL game. 

\subsection{Server's Decision at Stage I}

In Stage I, the server chooses its pricing strategy based on all the clients' best responses. In other words, the server substitutes $q_n^*(P_n)$ into its utility function in \eqref{U_s} for obtaining the optimal price vector  $\boldsymbol{P}$ under the budget constraint $B$. 

However, since the analytical expression of $q_n(P_n)$ is complicated, it is difficult to obtain the optimal $\boldsymbol{P}$ in the server's decision problem. As $q_n^\ast (P_n)$ is unique, we can write its inverse function based on \eqref{U_i_analysis}, i.e., 
$P_n(q_n)= 2c_nq_n-{a_n^2G_n^2}/{q_n^2}$. Then, we substitute this expression of $P_n(q_n)$ into the server's utility function and solve Stage I problem in \eqref{U_s}. 
Therefore, with the obtained convergence bound and the constraints of $q_n$, we rewrite problem in \eqref{U_s} as 
the following Problem \textbf{P1}$^{\prime}$:
\begin{subequations}
\label{U_s2}
\begin{align}
\textbf{P1}^{\prime}: \ \min_{\boldsymbol{q}} \ &   
F^{*} + \frac{1}{R}\left(\alpha \sum_{n=1}^{N} \frac{(1-q_n)a_{n}^{2} G_{n}^{2}}{q_{n}}+\beta\right) \label{U_s2_a}\\
         \text{s.t.} 
          \ \  &\sum_{n=1}^N\left(2c_nq_n-\frac{\alpha}{R}\frac{v_na_n^2G_n^2}{q_n^2}\right) q_n\le B, \label{U_s2_b}\\
           &\  0\le q_n\le q_{n,\text{max}}. \label{U_s2_c}       
\end{align}
\end{subequations}

Although the objective function in \eqref{U_s2_a} is convex in $q_n$, the budget constraint in \eqref{U_s2_b} is not convex 
in $q_n$. 
Thus,  Problem \textbf{P1}$^{\prime}$ is \emph{non-convex}.
To efficiently solve Problem \textbf{P1}$^{\prime}$,  we define a new control variable 
\begin{equation}
    \label{new_M}
    M:=\sum\nolimits_{n=1}^Nc_nq_n^2,
\end{equation}
where 
$0 \le M \le\sum_{n=1}^Nc_n$.
Then, we rewrite Problem \textbf{P1}$^{\prime}$ as the following Problem \textbf{P1}$^{\prime\prime}$:\footnote{We omit constants $F^{*}$ and $\beta$ in the objective of \textbf{P1}$^{\prime}$ for simplicity.}
\begin{equation}
\label{p3}
\begin{aligned}
\!\!\!\! \textbf{P1}^{\prime\prime}:  \quad \min_{\boldsymbol{q},M} \    &g\left(\boldsymbol{q},M\right):=   \frac{\alpha}{R}\sum_{i=1}^{N} \frac{(1-q_n)a_{n}^{2} G_{n}^{2}}{q_{n}}\\
 \text{s.t.} \ \  \ \  &    2M-\sum_{n=1}^N\frac{\alpha}{R}
 \frac{v_na_n^2G_n^2}{q_n} \le B,  \\ &\sum\nolimits_{n=1}^Nc_nq_n^2=M, \ \ 0 \le q_n\le q_{n,\text{max}}. 
\end{aligned}
\end{equation}
For any fixed feasible value of $M$, Problem \textbf{P1}$^{\prime\prime}$ is convex because the objective function and the constraints are convex. 
Hence, we can approximately solve Problem \textbf{P1}$^{\prime\prime}$ in two steps. First, for any fixed $M$, we solve for the optimal  $\boldsymbol{q}^*(M)$ in Problem \textbf{P1}$^{\prime\prime}$ 
via a convex optimization tool, e.g., CVX \cite{boyd2004convex}.  This allows us to write the objective function of Problem \textbf{P1}$^{\prime\prime}$ as $g(\boldsymbol{q}^\ast(M), M)$. Then we will further solve the problem by using 
a linear search method  with a fixed step-size $\epsilon_0$ over the interval $\left[0, \sum_{n=1}^Nc_n{q_{n,\text{max}}^2}
\right]$, which leads to $\boldsymbol{q}^*(M^*(\epsilon_0))=\arg\min_{M(\epsilon_0)}\boldsymbol{q}^*(M(\epsilon_0))$. 


Once we obtain $\boldsymbol{q}^*$ from Problem \textbf{P1}$^{\prime\prime}$, we immediately have the optimal  price $\boldsymbol{P}^*$ via \eqref{U_i_analysis} as follows,
\begin{equation}
\label{optimal_P}
    P_n^*=2c_nq_n^*-v_n\frac{\alpha}{R}
    \frac{a_n^2G_n^2}{{q_n^*}^2}.
\end{equation}

Finally, 
we conclude that the obtained solution pair $\{\boldsymbol{q}^*, \boldsymbol{P}^*\}$ based on backward induction are the SE $\{\boldsymbol{P}^{{\text{SE}}},\boldsymbol{q}^{{\text{SE}}}\}$ for the proposed CPL game \cite{bacsar1998dynamic}. 
Notably, based on the relationship between $q_n^{{\text{SE}}}$ and $P_n^{{\text{SE}}}$ in \eqref{optimal_P}, we highlight that 
clients with low participation level $q_n^{{\text{SE}}}$ will receive a low price $P_n^{{\text{SE}}}$, which leads to a low payment  $P_n^{{\text{SE}}}q_n^{{\text{SE}}}$ from the server. In particular, if these clients with low participation levels also have high intrinsic value $v_n$,  they may need to pay for the server, i.e., $P_n^{\text{SE}}<0$, which we show in the next section.


\subsection{Properties of SE in the CPL Game}

This subsection presents some interesting properties of the obtained SE  $\{\boldsymbol{P}^{\text{SE}},\boldsymbol{q}^{\text{SE}}\}$. 
Before that, we first give the following lemma which leads to our property analysis.
\begin{lemma}
\label{ineq_to_eq}
At the SE of the CPL Game, the server's budget constraint is tight. 
\end{lemma}
{\setlength{\parskip}{0.4em}
\noindent\emph{Proof  Sketch.}
 The idea of this proof is to use contradiction, which we omit due to page limitation. 
\qquad \qquad \qquad  \qquad \qedsymbol}

\vspace{1mm}
\subsubsection{\textbf{Impact of the server's budget $B$}} 
We first show the impact of the server's budget $B$ on SE  $\{\boldsymbol{P}^{\text{SE}},\boldsymbol{q}^{\text{SE}}\}$. 
\newtheorem{Proposition}{{Proposition}} \begin{Proposition}
\label{property_B}
 Both $\boldsymbol{P}^{\text{SE}}$ and $\boldsymbol{q}^{\text{SE}}$ increase in  
 budget $B$.   
 \end{Proposition}

\begin{proof}
Given $B\!=\!\sum\limits_{n=1}^N\left(2c_n{q_n^*}^2-\!\frac{\alpha}{R}\frac{v_na_n^2G_n^2}{q_n^*}\right)$ from Lemma~\ref{ineq_to_eq}, we have ${\partial B}/{\partial q_n^*}=4c_nq_n^*+\frac{\alpha}{R}\frac{v_na_n^2G_n^2}{q_n^*}^2>0$. Thus, $q_n^*$ increases in $B$. Next, based on the {monotonically increasing} relationship between $q_n^*$ and $P_n^*$ in \eqref{optimal_P}, we conclude that $P_n^*$ also increases in $B$.     
\end{proof}

Proposition~\ref{property_B} shows that when the budget $B$ is higher,  the server could increase its price $P_n$ to  incentivize higher client participation level $q_n$ to decrease the global loss for improving the model performance. Similarly, clients also have incentive to increase their participation level for more profit. 

{\setlength{\parskip}{0.3em}\subsubsection{\textbf{Impact of clients' parameters}}
This subsection shows how client's parameters, including data quality, local cost, and intrinsic value affect the values of  $\{\boldsymbol{P}^{\text{SE}},\boldsymbol{q}^{\text{SE}}\}$ at SE. 
For simplicity, we consider those clients whose equilibrium choices are in the (strict) interior of domains, i.e., $0< q_n^{\text{SE}} < q_{n,\text{max}}$. 
}

\begin{theorem}
\label{property_q}
{(Impact of clients' parameters on  $\boldsymbol{q}^{\text{SE}}$)} 
For any 
clients $i$ and $j$ whose equilibrium $q_i^{\text{SE}}$ and $q_j^{\text{SE}}$ 
are in the interior of the domains, we must have 
$\frac{c_i{q_i^*}^3}{a_i^2G_i^2}+v_i = \frac{c_j{q_j^*}^3}{a_j^2G_j^2}+v_j$. 


\end{theorem}
\begin{proof}
We give the proof in Appendix C. 
\end{proof}
We summarize the key insights of Theorem~\ref{property_q} as follows:
\begin{itemize}
 \item  \emph{(Heterogeneous Data quality)}  
 Clients with 
 larger  $a_nG_n$ (e.g., large data size and gradient norm upper-bound) 
 have higher participation level $q_n^{\text{SE}}$, given the same  parameters $c_n$ and  $v_n$ among clients.   
 

\item \emph{(Local cost)} 
Clients with large local cost parameter $c_n$  have lower participate level  $q_n^{\text{SE}}$, given the same parameters $a_nG_n$ and $v_n$ among clients. 

\item \emph{(Intrinsic value)} It is \emph{counter-intuitive} that a client with a larger intrinsic value parameter $v_n$ has a lower participation rate $q_n^{\text{SE}}$, given the same $a_nG_n$ and $c_n$ among clients. This is because although a higher $q_n$ benefits its intrinsic value if its $v_n$ is large, the server will set a lower price $P_n$ (as we show in the next theorem) which yields a smaller payment $P_nq_n$. Considering that the larger $q_n$ also incurs larger cost $c_nq_n^2$, $q_n$ should not be large. 


\end{itemize}

\begin{theorem}
\label{property_P}
{(Impact of  clients' parameters on $\boldsymbol{P}^{\text{SE}}$)} 
For any client $n$ whose equilibrium $q_n^{\text{SE}}$ is in the interior domain, 
we must have 
\begin{equation}
\label{optimal_P2}
    P_n^*\!=\!\left(\!\frac{2\alpha c_n^2a_n^2G_n^2}{R}\right)^{\!\frac{1}{3}}\left[\left(\frac{1}{\lambda^*}\!-\!v_n\right)^{\!\frac{1}{3}}\!\!-\!2\left(\frac{v_n^{\frac{3}{2}}}{\frac{1}{\lambda^*}\!-\!v_n}\!\right)^{\!\frac{2}{3}}\right]\!.
\end{equation}
In addition, there exists a 
threshold $v^{\text{t}}$, such that $P_n^{\text{SE}}>0$, if $v_n \le v^{\text{t}}$, and $P_n^{\text{SE}}<0$ otherwise.
\end{theorem}
\begin{proof}
We give the proof in Appendix D.
\end{proof}

Before we present the insights of Theorem~\ref{property_P}, we show how 
parameters $c_n$ and $a_nG_n$ affect  $P_n^{\text{SE}}$ with the following result. 
\newtheorem{Corollary}{{Corollary}}
\begin{Corollary}\label{corollary}
With the same threshold $v^{\text{t}}$ in Theorem~\ref{property_P}, and for any clients $i$ and $j$ whose equilibrium $q_i^{\text{SE}}$ and $q_j^{\text{SE}}$ are in the interior of the domains and satisfy  $c_ia_iG_i>c_ja_jG_j$, 
\begin{enumerate}
    \item 
 if $v_i<v_j<v^{\text{t}}$,  
then $P_i^{\text{SE}}>P_j^{\text{SE}}>0$;
\item  if $v_i>v_j>v^{\text{t}}$, 
then $P_i^{\text{SE}}<P_j^{\text{SE}}<0$.
\end{enumerate}
\end{Corollary}

We summarize the key insights of Theorem~\ref{property_P} and Corollary~\ref{corollary} as follows: 
\begin{itemize}
  \item \emph{(Intrinsic value)} Our result 
  provides a quantitative criterion, 
  which \emph{indicates the payment direction of $P_n^{\text{SE}}$ between the server and the clients}.  
    
 \item \emph{(Data quality)} 
 We have shown that, no matter which payment direction, the price $P_n^{\text{SE}}$ is higher for clients who have 
 large values of $a_nG_n$, given the same parameters $c_n$ and $v_n$ among  clients.  This result demonstrates the suboptimality of existing mechanisms with uniform pricing or data quantity-based  pricing. 
 
 
  \item \emph{(Local cost)} It is \emph{counter-intuitive} that clients with large  $c_n$ will have higher price $P_n^{\text{SE}}$ when parameters $a_nG_n$ and $v_n$ are the same as other clients.  
  However, this result makes sense because a client with a large local cost $c_n$ tends to join with low participation level $q_n$, which can cause a negative impact on the server's utility. Hence,  
  to prevent this, the server will set a higher price. 

       \end{itemize}

\section{Experimental Evaluation}
\label{sec:experimentation}
In this section, we empirically evaluate the performance of our proposed mechanism with different datasets on our hardware-based cross-device FL prototype, as illustrated in Fig.~\ref{hdset}.
Our prototype consists of $N=40$ Raspberry Pis serving as clients 
and a laptop computer acting as the central server. All devices are interconnected via an enterprise-grade Wi-Fi router. We develop a TCP-based socket interface for the communication between the server and clients. 

In the following, we first present the evaluation setup and then show the experimental results.

\begin{figure}[t]
	\centering
	\includegraphics[width=7.4cm,height=4.6cm]{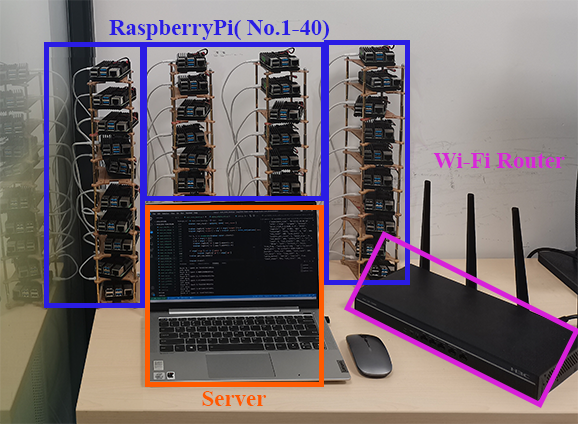}
	\caption{Cross-device FL prototype with a laptop serving as the central server and
$40$ Raspberry Pis serving as clients. 
	}
	\label{hdset}
\end{figure}

\begin{table}[ht]
\small
 \caption{System parameters for different Setups}
 \label{sys}
  \centering
  \small
    \begin{tabular}{c|c|c|c}
    \toprule
   {\!Setup$\backslash$\!\! Parameter} &  {budget $B$} &  local cost $\overline{c}$  &   intrinsic value $\overline{v}$    \bigstrut\\
    \hline
   {\makecell[c]{\textbf{Setup 1}}} &  ${200}$  & $50$   & $4,000$   
   \bigstrut\\
    \hline
   {\makecell[c]{\textbf{Setup 2}}}  & ${40}$  & $20$   & $30,000$  \bigstrut\\
    \hline
{\makecell[c]{\textbf{Setup 3}}} &  ${500}$  & $80$   & $10,000$   \bigstrut\\
    \bottomrule
    \end{tabular}%
   \label{tab1}%
\end{table}

\subsection{Experimental Setup}

\subsubsection{\textbf{Datasets and Implementations}} 
Following  similar setups as in \cite{li2019convergence,luo2022tackling}, we evaluate our results on three datasets, with detailed implementations as follows:

\begin{itemize}
    \item \textbf{Setup 1}: The first experiment uses the Synthetic dataset, which generates 60-dimensional random vectors as input
data with a  non-i.i.d.  $Synthetic \ (1, 1)$ setting. We generate $22,377$ data samples and distribute them among the devices in an \emph{unbalanced} power-law distribution.  
  
    \item   \textbf{Setup 2}: The second experiment uses MNIST dataset, where we randomly subsample 
$14,463$ data samples from MNIST and distribute them among the devices in an \emph{unbalanced} (following the power-law distribution) and \emph{non-i.i.d.}  (i.e., each device has $1$--$6$ classes) fashion.

     \item \textbf{Setup 3}: The third experiment uses EMNIST dataset, where we randomly subsample $35,155$ lower case character samples from the EMNIST dataset and distribute  among the devices in an \emph{unbalanced} (i.e., numbers of data samples at each device follows a  power-law distribution) and \emph{non-i.i.d.} fashion (i.e., each device has a randomly chosen number of classes, ranging from $1$ to $10$). 
\end{itemize}

\subsubsection{\textbf{Model and Parameters}} 
 For all experiments, we adopt the {convex} {multinomial logistic regression} model, 
 with $\boldsymbol{w}_0\!=\!\boldsymbol{0}$ and  SGD batch size $24$. We use an initial learning rate of $\eta_0 =0.1$, a decay rate of $0.996$, 
 and a local iteration number $E=100$. 
We estimate the task-related parameters $\alpha$ and data quality-related parameter $G_n$ for each setup following a similar approach as \cite{luo2022tackling}. 
We let $q_{n,max}=1$ for all $n$ and training round $R=1000$ for all three setups. Table~\ref{sys} shows the parameter settings for budget $B$, mean local cost parameter $\overline{c}$, and mean intrinsic value parameter $\overline{v}$ for each setup, with $\overline{c}$ and $\overline{v}$ following exponential distribution among clients. {We note that the parameters and results are general, and the specific parameters for a scenario can often be obtained through measurement in practice.} 

\begin{figure}[!t]
\centering
\subfigure[Loss for Setup 1 ]{\label{bench_loss1}\includegraphics[width=4.37cm, height=3.3cm]{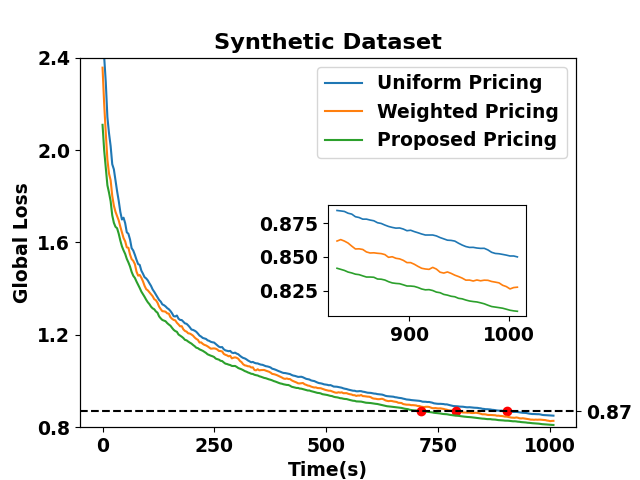}}
\subfigure[Accuracy for Setup 1 ]{\label{bench_acc1}\includegraphics[width=4.37cm, height=3.3cm]{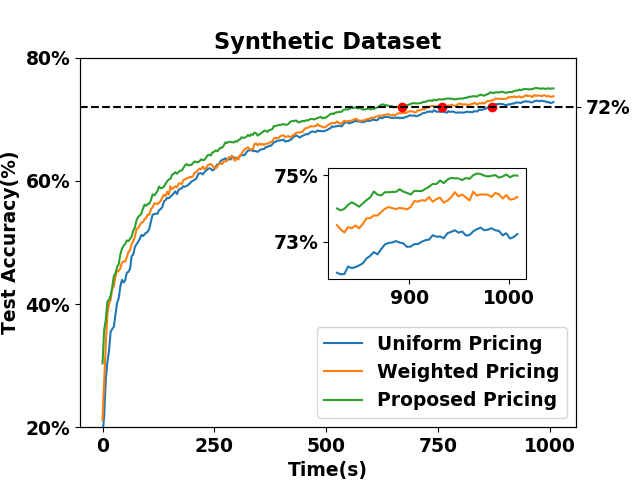}}
\subfigure[Loss for for Setup 2]{\label{bench_loss2}\includegraphics[width=4.37cm, height=3.3cm]{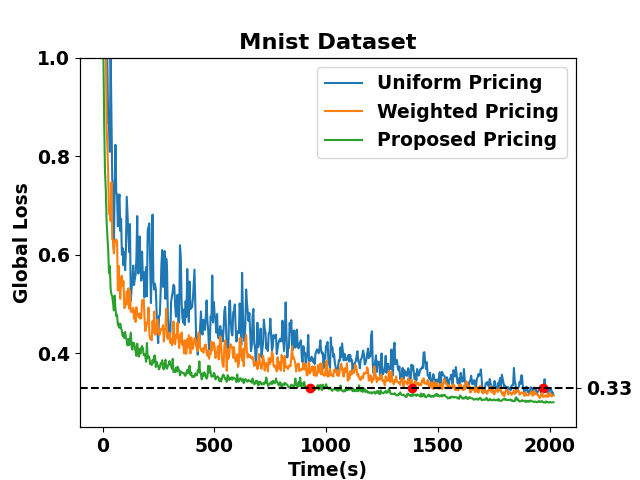}}
\subfigure[Accuracy for for Setup 2]{\label{bench_acc2}\includegraphics[width=4.37cm, height=3.3cm]{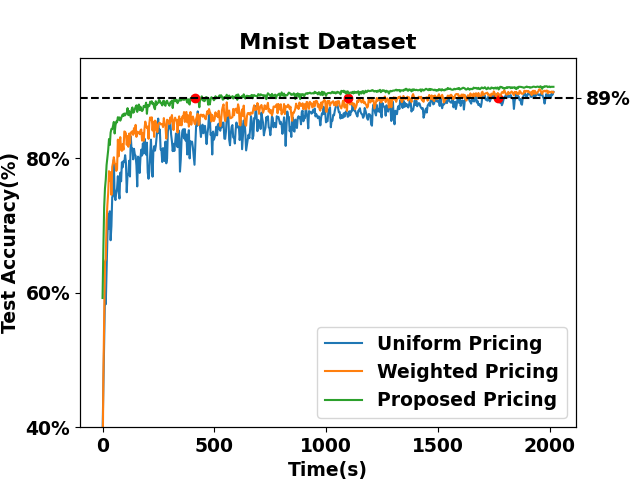}}
\subfigure[Loss for for Setup 3]{\label{bench_loss3}\includegraphics[width=4.37cm, height=3.3cm]{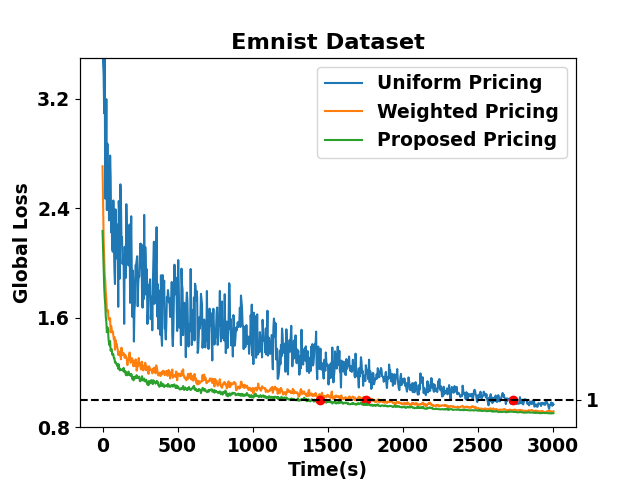}}
\subfigure[Accuracy for for Setup 3]{\label{bench_acc3}\includegraphics[width=4.37cm, height=3.3cm]{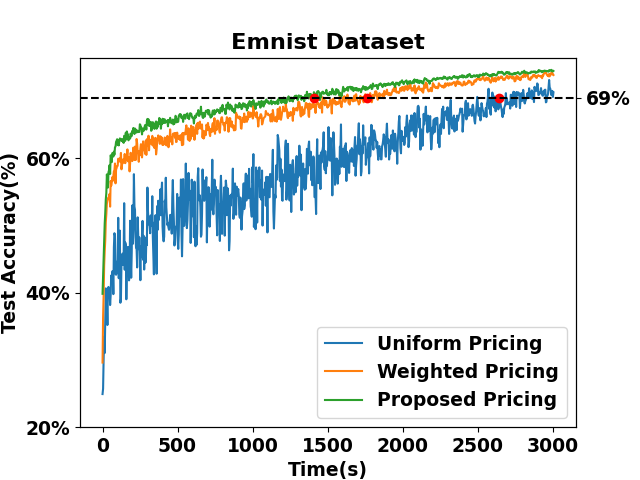}}
\caption{
Model performance of loss and accuracy for Setups $1$ -- $3$ for different pricing schemes.
}
\label{bench}
\vspace{-2mm}
\end{figure}

\subsection{Experimental Results}
We first compare the performance of our proposed optimal pricing with other benchmark \emph{pricing schemes}, and then we evaluate the \emph{impact of different system parameters} on the model performance and Equilibrium solution.  For all three setups, we average each experiment over $20$ independent runs. 

{\setlength{\parskip}{0.3em}\subsubsection{\textbf{Comparison with Different Pricing Schemes}} 
We compare our proposed optimal pricing $\boldsymbol{P}^*$ with two benchmark pricing schemes: \emph{uniform pricing} $\boldsymbol{P}^{\text{u}}$ where the server sets the same price for all devices, and \emph{weighted pricing} $\boldsymbol{P}^{\text{w}}$ where clients' prices are proportional to their datasize. 
Then, we run experiments for all three Setups using the corresponding optimal clients' participation levels $\boldsymbol{q}^*$, $\boldsymbol{q}^\text{u}$, and $\boldsymbol{q}^\text{w}$, based on the above three pricing schemes. 
Fig.~\ref{bench} illustrates the global model performance of global loss and test accuracy for different pricing schemes for Setups $1$--$3$ 
 with evaluating time $1,000$, $2,000$, and $3,000$ seconds, respectively. }
 
{\setlength{\parskip}{0.1em}\textbf{Loss.} As shown in Figs.~\ref{bench_loss1}, \ref{bench_loss2}, and \ref{bench_loss3},  our proposed optimal pricing scheme \emph{achieves lower global loss with smaller variance} throughout the evaluating process compared to the other benchmarks under the same budget. 
Specifically, for EMNIST dataset,  Fig.~\ref{bench_loss2} shows that our scheme reaches a target loss of $0.33$ using around $33.1\%$ less time than weighted pricing and around $52.9\%$ less time than the uniform pricing.  
We summarize the superior performances of loss in Table~\ref{table_loss}.} 

{\setlength{\parskip}{0.1em}\textbf{Accuracy.} Figs.~\ref{bench_acc1}, \ref{bench_acc2}, and \ref{bench_acc3} show that our proposed optimal pricing scheme \emph{achieves higher test accuracy with  smaller  variance} compared to the other benchmarks under the same budget. In particular, for MNIST dataset,  Fig.~\ref{bench_acc2} shows that our proposed pricing scheme spends around  $20.0\%$ less time  than the weighted pricing and around $46.5\%$ less time than the uniform pricing for reaching a target accuracy of $69\%$.  We summarize the superior performances of accuracy in Table~\ref{table_acc}.}

{\setlength{\parskip}{0.1em}\textbf{Clients' Utility.} Table~\ref{bench_un} shows the gain of clients' total utility of our proposed pricing  over the other two benchmarks. 
 We see that our proposed optimal pricing scheme \emph{yields higher total clients' utility} compared to the other two benchmark pricing schemes.} 

The above observations validate that our designed mechanism with optimal pricing can \emph{incentivize high-quality clients to participate with higher participation levels}  under the same budget, which benefits both the server and clients' utilities. 


\begin{table}[!t]
 \caption{Running Time for reaching the target loss (shown in Fig.~$4$) 
 for different pricing schemes 
 }
  \centering
    \begin{tabular}{c|c|c|c}
    \toprule
   {Setup $\backslash$ Pricing Schemes} &  \textbf{Proposed} &  Weighted  &   Uniform    \bigstrut\\
    \hline
   {\makecell[c]{\textbf{Setup 1}}} &  $\mathbf{711}$ \textbf{s} & $791$ s  & $903$ s  
   \bigstrut\\
    \hline
   {\makecell[c]{\textbf{Setup 2}}}  & $\mathbf{926}$ \textbf{s}  & $1,384$ s   & $1,969$ s \bigstrut\\
    \hline
{\makecell[c]{\textbf{Setup 3}}} &  $\mathbf{1,448}$ \textbf{s} & $1,758$ s   & $2,735$  s \bigstrut\\
    \bottomrule
    \end{tabular}%

 \label{table_loss}%
    \vspace{-1mm}
\end{table}

\begin{table}[!t]

 \caption{Running time for reaching the target accuracy (shown in Fig.~$4$) 
 for different pricing schemes
 }
  \centering
    \begin{tabular}{c|c|c|c}
    \toprule
   {Setup $\backslash$ \!Pricing Schemes} &  \textbf{Proposed} &  Weighted  &   Uniform    \bigstrut\\
    \hline
   {\makecell[c]{\textbf{Setup 1}}} &  $\mathbf{668}$ \textbf{s} & $759$ s  & $871$ s  
   \bigstrut\\
    \hline
   {\makecell[c]{\textbf{Setup 2}}}  & $\mathbf{411}$ \textbf{s} & $1,096$ s  & $1,767$ s \bigstrut\\
    \hline
{\makecell[c]{\textbf{Setup 3}}} &  $\mathbf{1,412}$ \textbf{s} & $1,766$ s  & $2,645$ s  \bigstrut\\
    \bottomrule
    \end{tabular}%
  \label{table_acc}%
    \vspace{-1mm}
\end{table}

\begin{table}[!t]
 \caption{Total clients' utility gain of our proposed pricing scheme over the other two benchmark pricing schemes}
 \label{bench_Un}
  \centering
    \begin{tabular}{c|c|c}
    \toprule
   {Setup $\backslash$ \! Gain} &  $\sum\limits_{n=1}^N U_n^*-\sum\limits_{n=1}^N U_n^{\text{u}}$ &  $\sum\limits_{n=1}^N U_n^*-\sum\limits_{n=1}^N U_n^{\text{w}}$    \bigstrut\\
    \hline
   {\makecell[c]{\textbf{Setup 1}}} &  ${10,839}$  & $5,175$
   \bigstrut\\
    \hline
   {\makecell[c]{\textbf{Setup 2}}}  & ${77,975}$  & $75,909$    \bigstrut\\
    \hline
{\makecell[c]{\textbf{Setup 3}}} &  ${86,200}$  & $22,413$         \bigstrut\\
    \bottomrule
    \end{tabular}%
  
  \label{bench_un}%
    \vspace{-2mm}
\end{table}

{\subsubsection{\textbf{Impact of System Parameters}}
Fig.~\ref{property_v} - Fig.~\ref{fig_property_B} show the impact of system parameters on the global model performance of our proposed mechanism  with evaluating time $600$ seconds. Due to page limit, we evaluate each   parameter with one setup.} 

{\textbf{Impact of $\overline{v}$}: Fig.~\ref{property_v} shows that as clients' mean intrinsic value $\overline{v}$ increases in Setup $1$, the obtained model achieves a lower  loss and a higher accuracy, respectively.    This observation is because when clients have more interest in the global model, they have internal motivation to participate in higher levels. 
In addition, as predicted by our theory,  we show in Table~\ref{table_property_P} that as $\overline{v}$ increases, the number of clients whose payment is negative also increases, e.g., from $0$ to $5$.}

{\textbf{Impact of $\overline{c}$}: Fig.~\ref{property_c} shows that as clients' mean local cost $\overline{c}$ decreases in Setup $2$, the obtained model achieves a lower  loss, a higher accuracy, and a  smaller variance. This  is because large cost prevents  clients from participating in higher levels.} 

{\textbf{Impact of $B$}:  Fig.~\ref{fig_property_B} shows that as budget $B$ increases in Setup $3$, the obtained  model achieves a lower loss, a higher  accuracy, and a  smaller variance.  This observation is because more budget allows more clients to participate in higher levels.}

\begin{table}[!t]
 \caption{Number of negative payment clients for different  $\overline{v}$ }
  \centering
    \begin{tabular}{c|c|c|c}
    \toprule
    {Setup $1$  (Synthetic)} &  $\overline{v}=0$  & $\overline{v}=4,000$ & $\overline{v}=80,000$
   \bigstrut\\
   \hline
Client number with $P_n<0$  & $0$  & $3$ & $5$
   \bigstrut\\
    \bottomrule
    \end{tabular}%
  
  \label{table_property_P}%
\end{table}

\begin{figure}[!t]
\centering
\subfigure[Loss  for different $\overline{v}$]{\label{property_v_loss}\includegraphics[width=4.37cm, height=3.3cm]{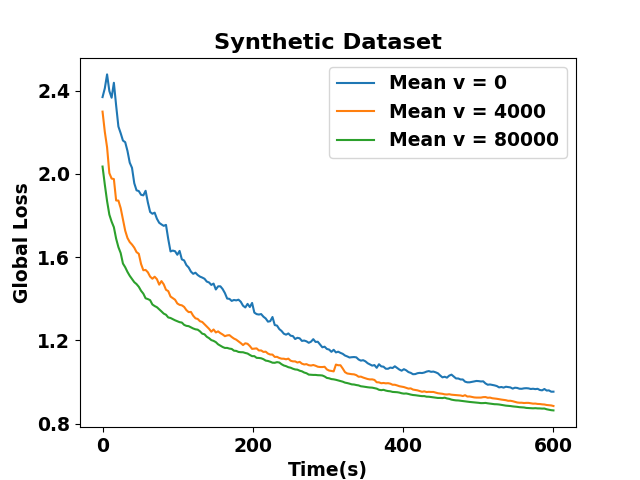}}
\subfigure[Accuracy for different $\overline{v}$]{\label{property_v_acc}\includegraphics[width=4.37cm,height=3.3cm]{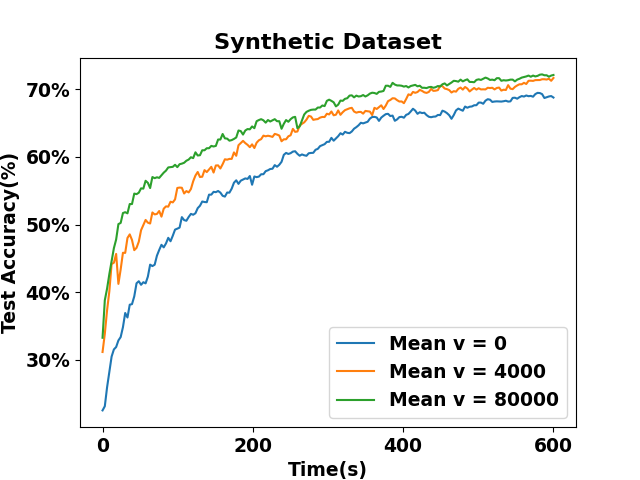}}
\caption{
Model performance  for different 
$\overline{v}$ for Setup 1.
}
\label{property_v}
\vspace{-1mm}
\end{figure}

\begin{figure}[!t]
\centering
\subfigure[Loss for different $\overline{c}$]{\label{property_c_loss}\includegraphics[width=4.37cm,height=3.3cm]{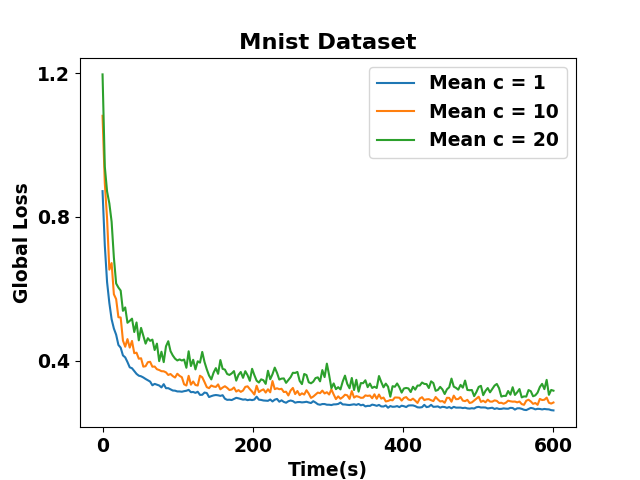}}
\subfigure[Accuracy  for different $\overline{c}$]{\label{property_c_acc}\includegraphics[width=4.37cm,height=3.3cm]{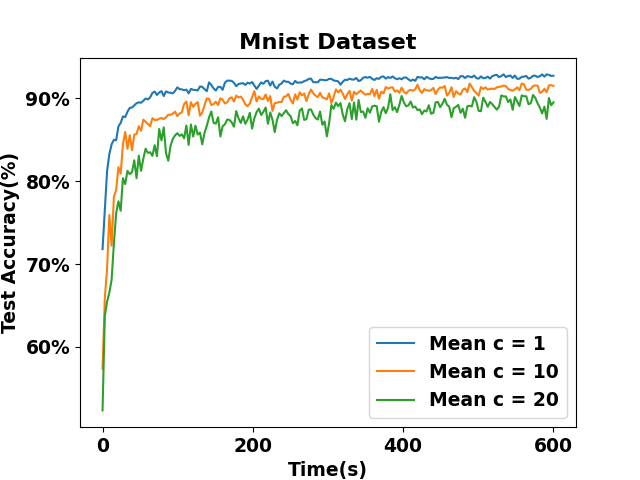}}
\caption{
Model performance for different 
$\overline{c}$ for Setup 2. 
}
\label{property_c}
\vspace{-1mm}
\end{figure}

\begin{figure}[!t]
\centering
\subfigure[Loss for different $B$]{\label{fig_property_B_loss}\includegraphics[width=4.37cm, height=3.3cm]{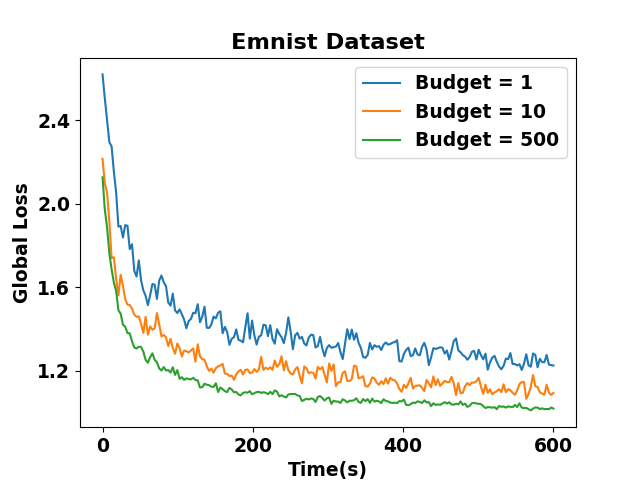}}
\subfigure[Accuracy for different $B$]{\label{fig_property_B_acc}\includegraphics[width=4.37cm, height=3.3cm]{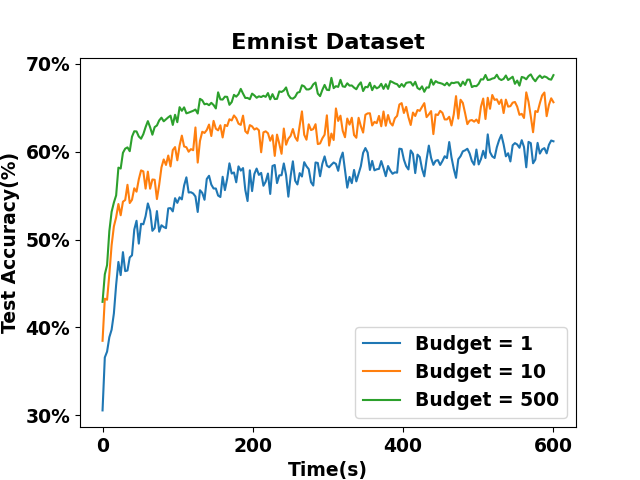}}
\caption{
Model performance for different $B$ for Setup 3.
}
\label{fig_property_B}
\vspace{-2mm}
\end{figure}


\section{Conclusion and Future Work}
\label{sec:conclusion}
In this work, we proposed a randomized client participation mechanism for FL, which guarantees  that the obtained model is unbiased and converges to the globally optimal model. 
We derived a new tractable
convergence bound that analytically  characterizes the impact of clients' participation levels and their non-i.i.d data on the server’s model performance, which yields a more accurate and customized pricing scheme.  To characterize clients' internal interests in the global model, we introduced and modeled a new \emph{intrinsic value} in clients' utilities, which allows clients to pay for the server.  We showed both theoretical and empirical evidence that the intrinsic value plays a critical role in the server and clients' decision-making. 
We conducted extensive experiments 
on a developed hardware prototype to demonstrate the superiority of our mechanism. 

Our mechanism serves as an initial step towards the incentive mechanism design to achieve an  unbiased and convergence guaranteed model with practical randomized client participation. 
In the future, we will extend our incentive mechanism for \emph{incomplete information} scenarios using \emph{Bayesian} method. 
We will further refine our cost model by decoupling the local cost into computation and communication 
consumption, and design an effective measurement to model clients' intrinsic value. 


\appendix
\subsection{Proof of  Lemma 1}
Substituting \eqref{aggregation} into \eqref{unbiased_agg}, we have
\begin{equation}
\label{prooflemma2}
\begin{array}{cl}
\Expect_{\mathcal{S}(q)^{r}}\!\!\left[\mathbf{w}^{r+1}\right]\!\!\!\!\!\!&=\!\mathbf{w}^{r}\!+\!\Expect\left[\sum_{n \in \mathcal{S}(q)^{r}} \frac{a_{n}}{q_{n}} \left(\mathbf{w}^{r+1}_{n}\!-\!\mathbf{w}^{r}\right)\right]
\\

&=\mathbf{w}^r\!+\!\sum_{n=1}^Nq_n\frac{a_n}{q_n}\left(\mathbf{w}^{r+1}_{n}\!-\!\mathbf{w}^{r}\right)\\
&=\mathbf{w}^r\!+\!\sum_{n=1}^Na_n\left(\mathbf{w}^{r+1}_{n}\!-\!\mathbf{w}^{r}\right)\\
&=\mathbf{w}^r\!+\!\overline{\mathbf{w}}^{r+1}-\mathbf{w}^{r}=\overline{\mathbf{w}}^{r+1},
\end{array}
\end{equation}
where  $n$ in the first equation is a random variable in the randomly sampled set $\mathcal{S}(q)^{(r)}$, which concludes the proof.

\subsection{Proof Sketch of Theorem 1}
Following a similar proof of convergence under full client participation \cite{li2019convergence,stich2018local}, we first show the convergence result under full client participation is   $\mathbb{E}[F\left(\overline{\boldsymbol{w}}^R)\right)]\!-\!F^{*} \le{\beta}/{R}$, 
where $\mathbb{E}[F\left(\overline{\boldsymbol{w}}^R)\right)]$ is the  expected global loss after $R$ rounds with full participation, and $\beta$ is the same as in \eqref{convergence}.  
Then, we use mathematical induction to obtain a non-recursive bound on $\mathbb{E}_{\mathcal{S}(\boldsymbol{q})^{r}}\left\|\boldsymbol{w}^R-{\boldsymbol{w}}^*\right\|^{2}$, and show that its difference compared to the bound of full participation $\mathbb{E}\left\|\overline{\boldsymbol{w}}^R-{\boldsymbol{w}}^*\right\|^{2}$ is the  variance introduced in \eqref{variance}. 
After that, we converted the bound of $\mathbb{E}_{\mathcal{S}(\boldsymbol{q})^{r}}\left\|\boldsymbol{w}^R-{\boldsymbol{w}}^*\right\|^{2}$ to   $\mathbb{E}[F\left(\boldsymbol{w}^R(\boldsymbol{q})\right)]-F^{*}$ using $L$-smoothness and strong convexity of $F(\cdot)$, which yields the additional term of $\alpha{\sum\nolimits_{n=1}^N\frac{\left(1-q_n\right)a_n^2G_n^{2}}{q_n}}$ in \eqref{convergence} compared to the upper bound with full client participation. $\hfill\square$ 


 \subsection{Proof of  Theorem 2}
We write the Lagrangian function $L(\boldsymbol{q}, \lambda, {\mu_n}, {\gamma_n})$ of  \textbf{P1}$^{\prime}$ as 
\begin{equation}
\begin{aligned}
\label{Lag}
    L
    \!=\!F^{*} \!+\! \frac{1}{R}\!\left(\alpha \sum_{n=1}^{N} \frac{(1\!-\!q_n)a_{n}^{2} G_{n}^{2}}{q_{n}}\!+\!\beta\right)\!+\sum_{n=1}^N \gamma_n(q_n\!-\!q_{n,\text{max}})\\
    -\!\sum\nolimits_{n=1}^N \mu_nq_n +\lambda\left(\sum\nolimits_{n=1}^N\left(2c_nq_n^2-\frac{\alpha}{R}\frac{v_na_n^2G_n^2}{q_n}\right)\!-\!B\right),\\
\end{aligned}
\end{equation}
where $\lambda>0$, ${\mu_n}\ge0$ and $\gamma_n\ge0$ are Lagrangian multipliers (dual variables). Although Problem \textbf{P1}$^{\prime}$ is non-convex, the Karush-Kuhn-Tucker (KKT) conditions are \emph{necessary for
optimality}.
Thus, 
for the first order condition, we must have 
\begin{equation}
    \label{first_order}
   -\frac{\alpha}{R}\frac{a_n^2G_n^2}{{q_n^{*2}}}+\lambda^*\left(4c_nq_n^{*}+\frac{\alpha}{R}\frac{v_na_n^2G_n^2}{q_n^{*2}}\right)-\mu_n^*+\gamma_n^*=0.
\end{equation}

Given that $0< q_n^*< q_{n,\text{max}}$, we have $\mu_n^*=0$ and $\gamma_n^*=0$ due to the \emph{Complementary Slackness} in the KKT conditions, which we 
rewrite \eqref{first_order} as follows 
\begin{equation}
    \label{first_order2}
    \frac{1}{\lambda^*}=\frac{4R}{\alpha}\frac{c_n{q_n^*}^3}{a_n^2G_n^2}+v_n.
\end{equation}
The result is obtained since ${1}/{\lambda^*}$ 
is independent of index $n$.
\qed 

\subsection{Proof of Theorem 3}
We exploit the relationship between $q_n^*$ and $P_n^*$ in \eqref{optimal_P}, and substitute it 
into  
the first order condition  \eqref{first_order2}, which leads to \eqref{optimal_P2}.
By letting $P_n^*=0$, we have 
$v_n={1}/{(3\lambda^*)}$. Therefore, given that $P_n^*$ in \eqref{optimal_P2}  decreases in $v_n$, we must have  $P_n^*>0$ when $v_n<{1}/{(3\lambda^*)}$, and  $P_n^*<0$ when $v_n>{1}/{(3\lambda^*)}$. Hence, we conclude the proof by letting the threshold $v^{\text{t}}={1}/{(3\lambda^*)}$. \qed


\bibliographystyle{IEEEtran}
\bibliography{mybibfile}

\begin{thebibliography}{10}
\providecommand{\url}[1]{#1}
\csname url@samestyle\endcsname
\providecommand{\newblock}{\relax}
\providecommand{\bibinfo}[2]{#2}
\providecommand{\BIBentrySTDinterwordspacing}{\spaceskip=0pt\relax}
\providecommand{\BIBentryALTinterwordstretchfactor}{4}
\providecommand{\BIBentryALTinterwordspacing}{\spaceskip=\fontdimen2\font plus
\BIBentryALTinterwordstretchfactor\fontdimen3\font minus
  \fontdimen4\font\relax}
\providecommand{\BIBforeignlanguage}[2]{{%
\expandafter\ifx\csname l@#1\endcsname\relax
\typeout{** WARNING: IEEEtran.bst: No hyphenation pattern has been}%
\typeout{** loaded for the language `#1'. Using the pattern for}%
\typeout{** the default language instead.}%
\else
\language=\csname l@#1\endcsname
\fi
#2}}
\providecommand{\BIBdecl}{\relax}
\BIBdecl

\bibitem{kairouz2019advances}
P.~Kairouz, H.~B. McMahan, B.~Avent, A.~Bellet, M.~Bennis, A.~N. Bhagoji,
  K.~Bonawitz, Z.~Charles, G.~Cormode, R.~Cummings \emph{et~al.}, ``Advances
  and open problems in federated learning,'' \emph{arXiv preprint,
  arXiv:1912.04977}, 2019.

\bibitem{mcmahan2017communication}
B.~McMahan, E.~Moore, D.~Ramage, S.~Hampson, and B.~A. y~Arcas,
  ``Communication-efficient learning of deep networks from decentralized
  data,'' in \emph{Proceedings of the 20th International Conference on
  Artificial Intelligence and Statistics}, 2017, pp. 1273--1282.

\bibitem{zhan2021survey}
Y.~Zhan, J.~Zhang, Z.~Hong, L.~Wu, P.~Li, and S.~Guo, ``A survey of incentive
  mechanism design for federated learning,'' \emph{IEEE Transactions on
  Emerging Topics in Computing}, doi: 10.1109/TETC.2021.3063517.

\bibitem{zeng2021comprehensive}
R.~Zeng, C.~Zeng, X.~Wang, B.~Li, and X.~Chu, ``A comprehensive survey of
  incentive mechanism for federated learning,'' \emph{arXiv preprint,
  arXiv:2106.15406}, 2021.

\bibitem{zhang2021jsac}
M.~Zhang, E.~Wei, and R.~Berry, ``Faithful edge federated learning: Scalability
  and privacy,'' \emph{IEEE Journal on Selected Areas in Communications},
  vol.~39, no.~12, pp. 3790--3804, 2021.

\bibitem{li2018federated}
T.~Li, A.~K. Sahu, M.~Zaheer, M.~Sanjabi, A.~Talwalkar, and V.~Smith,
  ``Federated optimization in heterogeneous networks,'' in \emph{Proceedings of
  Machine Learning and Systems (MLSys)}, 2020.

\bibitem{lim2020hierarchical}
W.~Y.~B. Lim, Z.~Xiong, C.~Miao, D.~Niyato, Q.~Yang, C.~Leung, and H.~V. Poor,
  ``Hierarchical incentive mechanism design for federated machine learning in
  mobile networks,'' \emph{IEEE Internet of Things Journal}, vol.~7, no.~10,
  pp. 9575--9588, 2020.

\bibitem{kang2019incentive}
J.~Kang, Z.~Xiong, D.~Niyato, S.~Xie, and J.~Zhang, ``Incentive mechanism for
  reliable federated learning: A joint optimization approach to combining
  reputation and contract theory,'' \emph{IEEE Internet of Things Journal},
  vol.~6, no.~6, pp. 10\,700--10\,714, 2019.

\bibitem{ding2020optimal}
N.~Ding, Z.~Fang, and J.~Huang, ``Optimal contract design for efficient
  federated learning with multi-dimensional private information,'' \emph{IEEE
  Journal on Selected Areas in Communications}, vol.~39, no.~1, pp. 186--200,
  2020.

\bibitem{ng2020joint}
J.~S. Ng, W.~Y.~B. Lim, H.-N. Dai, Z.~Xiong, J.~Huang, D.~Niyato, X.-S. Hua,
  C.~Leung, and C.~Miao, ``Joint auction-coalition formation framework for
  communication-efficient federated learning in uav-enabled internet of
  vehicles,'' \emph{IEEE Transactions on Intelligent Transportation Systems},
  vol.~22, no.~4, pp. 2326--2344, 2020.

\bibitem{le2021incentive}
T.~H.~T. Le, N.~H. Tran, Y.~K. Tun, M.~N. Nguyen, S.~R. Pandey, Z.~Han, and
  C.~S. Hong, ``An incentive mechanism for federated learning in wireless
  cellular networks: An auction approach,'' \emph{IEEE Transactions on Wireless
  Communications}, vol.~20, no.~8, pp. 4874--4887, 2021.

\bibitem{jiao2020toward}
Y.~Jiao, P.~Wang, D.~Niyato, B.~Lin, and D.~I. Kim, ``Toward an automated
  auction framework for wireless federated learning services market,''
  \emph{IEEE Transactions on Mobile Computing}, vol.~20, no.~10, pp.
  3034--3048, 2020.

\bibitem{zeng2020fmore}
R.~Zeng, S.~Zhang, J.~Wang, and X.~Chu, ``Fmore: An incentive scheme of
  multi-dimensional auction for federated learning in mec,'' in
  \emph{Proceedings of the IEEE International Conference on Distributed
  Computing Systems (ICDCS)}, 2020, pp. 278--288.

\bibitem{deng2021fair}
Y.~Deng, F.~Lyu, J.~Ren, Y.-C. Chen, P.~Yang, Y.~Zhou, and Y.~Zhang, ``Fair:
  Quality-aware federated learning with precise user incentive and model
  aggregation,'' in \emph{Proceedings of the IEEE Conference on Computer
  Communications (INFOCOM)}, 2021, pp. 1--10.

\bibitem{yang2021achieving}
H.~Yang, M.~Fang, and J.~Liu, ``Achieving linear speedup with partial worker
  participation in non-iid federated learning,'' \emph{arXiv preprint,
  arXiv:2101.11203}, 2021.

\bibitem{li2019convergence}
X.~Li, K.~Huang, W.~Yang, S.~Wang, and Z.~Zhang, ``On the convergence of fedavg
  on non-iid data,'' in \emph{Proceedings of the International Conference on
  Learning Representation (ICLR)}, 2019.

\bibitem{qu2020federated}
Z.~Qu, K.~Lin, J.~Kalagnanam, Z.~Li, J.~Zhou, and Z.~Zhou, ``Federated
  learning's blessing: Fedavg has linear speedup,'' \emph{arXiv preprint,
  arXiv:2007.05690}, 2020.

\bibitem{chen2020optimal}
W.~Chen, S.~Horvath, and P.~Richtarik, ``Optimal client sampling for federated
  learning,'' \emph{arXiv preprint, arXiv:2010.13723}, 2020.

\bibitem{rizk2020federated}
E.~Rizk, S.~Vlaski, and A.~H. Sayed, ``Federated learning under importance
  sampling,'' \emph{arXiv preprint, arXiv:2012.07383}, 2020.

\bibitem{9579038}
B.~Luo, X.~Li, S.~Wang, J.~Huang, and L.~Tassiulas, ``Cost-effective federated
  learning in mobile edge networks,'' \emph{IEEE Journal on Selected Areas in
  Communications}, vol.~39, no.~12, pp. 3606--3621, 2021.

\bibitem{cho2020client}
Y.~J. Cho, J.~Wang, and G.~Joshi, ``Client selection in federated learning:
  Convergence analysis and power-of-choice selection strategies,'' \emph{arXiv
  preprint, arXiv:2010.01243}, 2020.

\bibitem{luo2022tackling}
B.~Luo, W.~Xiao, S.~Wang, J.~Huang, and L.~Tassiulas, ``Tackling system and
  statistical heterogeneity for federated learning with adaptive client
  sampling,'' in \emph{Proceedings of the IEEE Conference on Computer
  Communications (INFOCOM)}, 2022, pp. 1739--1748.

\bibitem{tu2021incentive}
X.~Tu, K.~Zhu, N.~C. Luong, D.~Niyato, Y.~Zhang, and J.~Li, ``Incentive
  mechanisms for federated learning: From economic and game theoretic
  perspective,'' \emph{arXiv preprint, arXiv:2111.11850}, 2021.

\bibitem{sarikaya2019motivating}
Y.~Sarikaya and O.~Ercetin, ``Motivating workers in federated learning: A
  stackelberg game perspective,'' \emph{IEEE Networking Letters}, vol.~2,
  no.~1, pp. 23--27, 2019.

\bibitem{khan2020federated}
L.~U. Khan, S.~R. Pandey, N.~H. Tran, W.~Saad, Z.~Han, M.~N. Nguyen, and C.~S.
  Hong, ``Federated learning for edge networks: Resource optimization and
  incentive mechanism,'' \emph{IEEE Communications Magazine}, vol.~58, no.~10,
  pp. 88--93, 2020.

\bibitem{9488705}
M.~Tang and V.~W. Wong, ``An incentive mechanism for cross-silo federated
  learning: A public goods perspective,'' in \emph{Proceedings of the IEEE
  Conference on Computer Communications (INFOCOM)}, 2021, pp. 1--10.

\bibitem{zhan2020incentive}
Y.~Zhan and J.~Zhang, ``An incentive mechanism design for efficient edge
  learning by deep reinforcement learning approach,'' in \emph{Proceedings of
  the IEEE Conference on Computer Communications (INFOCOM)}, 2020, pp.
  2489--2498.

\bibitem{candogan2012optimal}
O.~Candogan, K.~Bimpikis, and A.~Ozdaglar, ``Optimal pricing in networks with
  externalities,'' \emph{Operations Research}, vol.~60, no.~4, pp. 883--905,
  2012.

\bibitem{bacsar1998dynamic}
T.~Ba{\c{s}}ar and G.~J. Olsder, \emph{Dynamic noncooperative game
  theory}.\hskip 1em plus 0.5em minus 0.4em\relax SIAM, 1998.

\bibitem{Microeconomic}
\BIBentryALTinterwordspacing
A.~Mas-Colell, M.~Whinston, and J.~Green, \emph{Microeconomic Theory}.\hskip
  1em plus 0.5em minus 0.4em\relax Oxford University Press, 1995. [Online].
  Available: \url{https://EconPapers.repec.org/RePEc:oxp:obooks:9780195102680}
\BIBentrySTDinterwordspacing

\bibitem{yu2018parallel}
H.~Yu, S.~Yang, and S.~Zhu, ``Parallel restarted {SGD} for non-convex
  optimization with faster convergence and less communication,'' in
  \emph{Proceedings of the AAAI Conference on Artificial Intelligence}, 2019.

\bibitem{stich2018local}
S.~U. Stich, ``Local {SGD} converges fast and communicates little,'' in
  \emph{Proceedings of the International Conference on Learning Representation
  (ICLR)}, 2018.

\bibitem{wang2019adaptive}
S.~Wang, T.~Tuor, T.~Salonidis, K.~K. Leung, C.~Makaya, T.~He, and K.~Chan,
  ``Adaptive federated learning in resource constrained edge computing
  systems,'' \emph{IEEE Journal on Selected Areas in Communications}, vol.~37,
  no.~6, pp. 1205--1221, 2019.

\bibitem{cui2022optimal}
L.~Cui, X.~Su, Y.~Zhou, and J.~Liu, ``Optimal rate adaption in federated
  learning with compressed communications,'' in \emph{Proceedings of the IEEE
  Conference on Computer Communications (INFOCOM)}, 2022, pp. 1459--1468.

\bibitem{shi2020device}
W.~Shi, S.~Zhou, and Z.~Niu, ``Device scheduling with fast convergence for
  wireless federated learning,'' in \emph{Proceedings of the IEEE International
  Conference on Communications (ICC)}, 2020, pp. 1--6.

\bibitem{perazzone2022communication}
J.~Perazzone, S.~Wang, M.~Ji, and K.~S. Chan, ``Communication-efficient device
  scheduling for federated learning using stochastic optimization,'' in
  \emph{Proceedings of the IEEE Conference on Computer Communications
  (INFOCOM)}, 2022, pp. 1449--1458.

\bibitem{luo2020cost}
B.~Luo, X.~Li, S.~Wang, J.~Huang, and L.~Tassiulas, ``Cost-effective federated
  learning design,'' in \emph{Proceedings of the IEEE Conference on Computer
  Communications (INFOCOM)}, 2021, pp. 1--10.

\bibitem{boyd2004convex}
S.~Boyd, S.~P. Boyd, and L.~Vandenberghe, \emph{Convex optimization}.\hskip 1em
  plus 0.5em minus 0.4em\relax Cambridge university press, 2004.

\end{thebibliography}

\end{document}